\documentclass[a4paper,12pt,reqno,twoside, flalign]{article}
\usepackage[margin=2.0cm]{geometry}
\usepackage{setspace}
\usepackage[normalem]{ulem}
\usepackage{amsmath, mathtools, calligra}
\usepackage{amsfonts}
\usepackage{amssymb}
\usepackage[svgnames, x11names]{xcolor}
\usepackage{graphicx}
\usepackage{array}
\usepackage{pictexwd}
\usepackage{fancybox}
\usepackage{natbib}
\usepackage{newcent}
\usepackage{enumitem}
\usepackage{tikz}
\usepackage{mathrsfs}
\usepackage{bm}
\usepackage{subcaption}
\usepackage[section]{placeins}
\usepackage[colorlinks=true, citecolor= FireBrick, linkcolor=FireBrick, urlcolor=black]{hyperref}
\setlist{nolistsep}
\newcolumntype{C}[1]{>{\centering\let\newline\\\arraybackslash\hspace{0pt}}m{#1}}
\usepackage{diagbox}
\usepackage{tikz}

\DeclareMathAlphabet{\mathpzc}{OT1}{pzc}{m}{it}

\newtheorem{theorem}{Theorem}
\newtheorem{proposition}{Proposition}

\newtheorem{corollary}{Corollary}
\newtheorem{lemma}{Lemma}

\newtheorem{definition}{Definition}

\newenvironment{proof}[1][]
      {\par\medbreak{\noindent\bfseries Proof#1\quad}}
      {\hfill $\blacksquare$\bigbreak}

\def \proba{\mathbb{P}}

\def \marg {{\rm{marg}}}

\def \beq {\begin{eqnarray*}}
\def \eeq {\end{eqnarray*}}

\usepackage{comment}

\usepackage[textwidth=30mm]{todonotes}

\usepackage{mathtools}

\begin{document}
	\title{Comparing Experiments in Discounted Problems \thanks{We thank audiences at the 2023 Warwick Economic Theory Workshop, 2023 SAET Conference, the 2024 Mediterranean Game Theory Symposium, Collegio Alberto-CEPR workshop, the University of Toronto, UCL, HEC, QMUL, and Institut Henri Poincar\'e for comments and suggestions. XV is a member of the Gruppo Nazionale per l'Analisi Matematica, la Probabilita \`e le loro Applicazioni. XV is part of the Italian MIUR PRIN PNRR project ``Supply Chain Disruptions, Financial Losses and Their Preventions'' (project P2022XT8C8) and the MIUR PRIN 2022 project ``Learning in Markets and Society'' (project 2022EKNE5K).}}
	\date{Latest version: \today \\ First version: June 9, 2023}
%	\thanks{ }
	\author{Ludovic Renou\thanks{QMUL and  CEPR, Miles End, E1 4NS, London, UK, lrenou.econ(at)gmail.com } \;\& Xavier Venel\thanks{LUISS, Viale Romania 32, 00197 Roma, Italy, xvenel(at)luiss.it}}
%	\address{Ludovic Renou, Queen Mary University of London and  CEPR, Miles End, E1 4NS, London, UK}
%	\email{lrenou.econ(at)gmail.com}
%	\author{Xavier Venel}
%	\address{zavier Venel, LUISS, Viale Romania 32, 00197 Roma, Italy}
%	\email{xvenel(at)luiss.it}

\maketitle

	\begin{abstract} This paper compares statistical experiments in discounted problems, ranging from the simplest ones, where the state is fixed and the flow of  information exogenous, to more complex ones, where the decision-maker controls the flow of information or the state changes over time.


	\end{abstract}

\newpage 
\section{Introduction}\label{sec:intro}
Suppose a decision-maker must choose between two sources of information before making repeated decisions in an uncertain environment. When can we definitively state that one source is superior to the other? Relatedly, individuals routinely consume information from multiple, distinctive sources. Informal sources, such as friends, colleagues, and social media, typically provide information faster than formal sources like newspapers, experts, and government agencies, but often at the cost of accuracy. Which combination of sources is best?  These questions echo the central theme of this paper: the comparison of statistical experiments in discounted problems.\medskip 

In repeated problems, comparing information sources introduces a novel trade-off, as not only the overall informativeness of a source matters, but also its flow. For instance, one source may perfectly resolve the uncertainty, but only in the distant future, while another may continuously reveal partial information without ever completely resolving the uncertainty. Intuitively, highly patient decision-makers would prefer the former, while highly impatient decision-makers would prefer the latter. However, it is less clear which source moderately patient or impatient decision-makers would favor. This paper develops  a method to rank these two sources, regardless of the decision-maker's impatience. \medskip    

We start with the simplest discounted problems, where there is a fixed payoff-relevant state and, at each period, the decision-maker receives a signal, which may depend on all past signals (but not on past decisions), and then makes a decision. Throughout, we assume that the decision-maker evaluates sequences of decisions according to the discounting criterion, that is, as the discounted sums of instantaneous payoffs. The simplest model assumes general, but fixed, discount factors. General discount factors are flexible enough to model geometric discounting, hyperbolic discounting, and many other forms of discounting. In later sections, we generalize the model in several directions, ranging from relaxing the assumption of a fixed discount factor to assuming that the decisions made influence the signals received,  through relaxing the assumption of a fixed payoff-relevant state, and many others. We view these generalizations as our main contribution.\footnote{We start with the simplest discounted problems as they provide the most natural bridge between static problems and more advanced discounted problems.} Economic applications motivate these generalizations. For instance, in a dynamic monopolistic pricing problem the prices the seller offers over time partly determine what the seller learns about the unknown demand, say because the quantities demanded at each price offered are observed.  As another instance, in investment problems, it is natural to assume that the values (the payoff-relevant states), say of stocks or bonds, change over time, perhaps due to changes in macroeconomic and political conditions. As yet another instance,  it is not unusual to consider sufficiently patient decision-makers, without fixing a  precise value of the discount rate. The generalizations we propose speak to these economic applications, and a few others. \medskip

Following \cite{blackwell51, blackwell53}, we say that the statistical experiment $f$ is more valuable than the statistical experiment $g$ if $f$ induces a higher ex-ante discounted payoff than $g$, in all decision problems. Note that the comparison is made from the ex-ante viewpoint. We maintain that assumption in most of the paper, but relax it in Section \ref{subsec:seq-compa}, where we compare experiments sequentially over time. In the simplest discounted problems, an experiment $f$ is a collection of stochastic kernels $(f_t)$, where $f_t: \Theta \times X^{t-1}\rightarrow X_t$ maps states and profiles of past signals to probabilities over current signals. Let $f^t: \Theta \rightarrow \Delta(X^t)$ be the $t$-period experiment induced by $f$. The first theorem we prove states 
that the experiment $f$ is more valuable than $g$ in the simplest discounting problems if, and only if, the mixture experiment $\sum_{t}\delta_t f^t$ is Blackwell sufficient for the mixture experiment $\sum_{t}\delta_t g^t$. We call this condition $\delta$-sufficiency. In a nutshell,, the condition says that the experiment $f$ needs to be more informative than $g$ on average, and is weaker than requiring that $f^t$ is more informative than $g^t$ at all periods.\footnote{We call  this stronger condition $\Delta$-sufficiency and prove that it ranks experiments in all separable problems, a larger class of problems than the class of discounted problems.}   The intuition is clear. The simplest discounted problems we consider are equivalent to  static problems. To see this, observe that when the experiment is 
the mixture  $\sum_t \delta_t f^t$, the decision-maker observes the signal $(t,x^t)$ with probability $\delta_t f^t(x^t|\theta)$ when the payoff-relevant state is $\theta$ and chooses $a$ with probability $\sigma(a|t,x^t)$, where $\sigma$ is the decision-maker strategy. Clearly, if the decision-maker follows the strategy $\widehat{\sigma}$ in the discounted problem, where $\widehat{\sigma}_t(a|x^t) := \sigma(a|t,x^t)$, the discounted probability of decision $a$ is equal to the probability of decision $a$ in the corresponding static problem,  when the experiment  is $\sum_t \delta_t f^t$.  \medskip
 
In more general discounting problems, the above equivalence breaks down in many ways. For instance, suppose that the decision-maker controls 
the flow of information through the decisions made. The most natural idea is to let the experiments depend \emph{directly} on the decisions made.  This is, however, problematic as changing the decision problem would then change the experiments we compare -- what we want is to compare two \emph{fixed} experiments. The solution we adopt is to introduce ancillary variables into the definition of an experiment and to have decisions influence these variables. This formalism allows us to change the decision problem the decision-maker faces, without changing the experiments. We call these variables \emph{informational controls}. For instance, when a firm chooses its pricing and marketing strategies (decisions), it partly controls the fraction and types of consumers made aware of its products (controls) and, ultimately, the realized demand (observations).  To compare two \emph{controlled} experiments, we extend them to two larger experiments, where at each period an extended experiment produces a new control and signal, as a function of past controls and signals and the state. We then compare these larger experiments as in the uncontrolled case. We hasten to stress that the extensions are not unique. Intuitively, each decision problem induces a pair of \emph{linked} extensions through its mapping from decisions to controls. As a consequence, we need to do the comparison for \emph{all} possible extensions. We illustrate the comparisons with the help of an example -- Section \ref{sec:example-controlled}. \medskip

 We close the introduction with a brief discussion of two closely related papers \cite{greenshtein96} and \cite{whitmeyer-williams-2024}.   A more extensive discussion can be found in Section \ref{sec:rel-litt}. \cite{greenshtein96} compares statistical experiments in sequential problems, where the decision-maker's payoff function over sequences of decisions is arbitrary. In particular, payoff functions do not have to be separable across periods, let alone to be discounted sums of  instantaneous payoffs. In all other aspects, his model coincides with our simplest model (fixed state, exogenous evolution of the signals, ex-ante  comparison, etc).  Since Greenshtein considers a larger class of problems, his ranking of statistical experiments is coarser than ours. In fact, it is even stronger than requiring that $f^t$ sufficient for $g^t$ at all periods. In a concurrent paper, \cite{whitmeyer-williams-2024} compare statistical experiments in discounted and separable problems in the setting of our simplest model (fixed state, exogenous evolution of the signals, ex-ante  comparison, etc). Their characterizations are naturally equivalent to ours, although we present ours in terms of garblings, while theirs are in terms of posterior beliefs. Unlike us, both    
 \cite{greenshtein96} and \cite{whitmeyer-williams-2024} do not extend their analysis beyond our simplest model.  \medskip
%  The paper is organized as follows. Section \ref{sec:example} illustrates the main idea with the help of a simple model. Section \ref{sec:model} presents the formal set-up, while  Section \ref{sec:main} presents our results. Finally, Section \ref{sec:rel-litt} offers a detailed review of the related literature. 
 
  A word of caution is in order. It is well-known that comparing experiments with respect to their performances in \emph{all} decision problems typically results in coarse rankings. For a stark example, suppose than $f$ and $g$ are the uniform distributions on $[\theta-1/3, \theta+1/3]$ and $[\theta -1/2, \theta+1/2]$, respectively ($\theta$ is the unknown parameter).  Intuitively, we expect $f$ to be more informative than $g$ since its dispersion around $\theta$ is smaller. Surprisingly, the two experiments are not comparable -- see \cite{lehmann-1988}. We are no different.\footnote{Since we are considering the class of discounted problems, a subset of the class of all dynamic decision problems, we obtain a finer ranking than the one of \cite{greenshtein96}. See Section \ref{sec:rel-litt} for an in-depth discussion.}

\section{An Introductory Example} \label{sec:example}

This section motivates our study with the help of a simple example. A decision-maker has to make a decision in each of two periods, labelled $t=1,2$. If the decision-maker chooses $a_t \in A$ at period $t$, the instantaneous payoff is $u(a_t,\theta)$,  when the state is $\theta$. The decision-maker evaluates the sequence of decisions $(a_1,a_2)$ as the (normalized) sum of their instantaneous payoffs $(1/2)u(a_1,\theta)+ (1/2)u(a_2,\theta)$.\medskip

Prior to making a decision, the decision-maker has the opportunity to receive some information from one of two possible sources of information, denoted $f$ and $g$, respectively.  The first  source $f$ provides no information in the first period, but fully reveals the state in the second period, while the second source $g$ provides some information in the first period, but no additional information in the second period. The central question this paper addresses is whether it is possible to rank these two sources of information (modeled as statistical experiments), independently of the decision problem $(A,u)$. \medskip 

The decision-maker faces a clear trade-off. If he opts to obtain information from $f$, he is able to calibrate perfectly his second-period choice to the state at the cost of being poorly calibrated in the first period. Alternatively, if he opts to obtain information from $g$, he is able to calibrate his decision to the state at both periods, but imperfectly so. To compare the two experiments, we need to compare their flow informativeness -- comparing the overall informativeness is not enough. Clearly, the source $f$ is  the most informative, but the question is whether this is sufficient to compensate for the delay. \medskip 

Comparing $f$ and $g$ rests on the simple idea that we can view them as two static experiments, which then can be compared in the usual way. Indeed, we can think of $f$ as the static experiment $\widehat{f}$, which produces the uninformative signal $\emptyset$ with probability $1/2$ and a perfectly informative signal with the complementarity probability. Similarly, we view the experiment $g$ as a static experiment in its own right.   For concreteness, suppose that the state takes two possible values, $\theta$ and $\theta'$, and consider the experiments:
\begin{table}[h!]
\parbox{.50\linewidth}{
\centering
\begin{tabular}{|c|c|c|} \hline
\backslashbox{$\Theta$}{$\{x_1\} \times X_2$} & $(x_1,x_2)$ & $(x_1,x_2')$ \\ \hline
$\theta$ & $1$ & $0$ \\ \hline
$\theta'$ & $0$ & $1$ \\ \hline
\end{tabular}\caption{The experiment $f$}}
\parbox{.50\linewidth}{
\centering
\begin{tabular}{|c|c|c|} \hline
\backslashbox{$\Theta$}{$Y_1 \times \{y_2\}$} & $(y_1,y_2)$ & $(y_1',y_2)$ \\ \hline
$\theta$ & $7/12$ & $5/12$ \\ \hline
$\theta'$ & $5/12$ & $7/12$ \\ \hline
\end{tabular} \caption{The experiment $g$}}
\end{table}

\FloatBarrier
\medskip

The experiment $f$ produces the signal $x_1$ with probability 1 in the first period, regardless of the state, and the signal $x_2$ (resp., $x_2'$) with probability 1 in the second period when the state is $\theta$ (resp., $\theta'$), whereas the experiment $g$   
produces the signal $y_1$ (resp., $y_1'$) with probability $7/12$ in the first period when the state is $\theta$ (resp., $\theta'$) and the signal $y_2$ with probability 1 in the second period, regardless of the state and the first-period signal.

 The modified (static) experiments $\widehat{f}$ and $\widehat{g}$ are:
\begin{table}[h!]
\parbox{.50\linewidth}{
\centering
\begin{tabular}{|c|c|c|c|} \hline
\backslashbox{$\Theta$}{$X$} & $x$ & $x'$ & $\emptyset$ \\ \hline
$\theta$ & $1/2$ & $0$ & $1/2$\\ \hline
$\theta'$ & $0$ & $1/2$ & $1/2$\\ \hline
\end{tabular}\caption{The modified experiment $\widehat{f}$} \label{Table:intro-modified-expe}}
\parbox{.50\linewidth}{
\centering
\begin{tabular}{|c|c|c|} \hline
\backslashbox{$\Theta$}{$Y$} & $y$ & $y'$ \\ \hline
$\theta$ & $7/12$ & $5/12$ \\ \hline
$\theta'$ & $5/12$ & $7/12$ \\ \hline
\end{tabular} \caption{The modified experiment $\widehat{g}$}}
\end{table}

\medskip

We can verify that $\widehat{f}$ is more informative than $\widehat{g}$ \citep{blackwell51, blackwell53} and, therefore, the decision-maker prefers $f$ to $g$, regardless of the decision problem to be repeated. Intuitively, we exploit the simple observation that discounted payoffs are the expectation of static payoffs  with  respect to the discounted probabilities of decisions. See Theorem \ref{th:fixed-discounting-uncontrolled} for a formal statement.\medskip 

We conclude this section with a word of caution. The example is deliberately simple and abstracts from many complications we will deal with. 
For instance, the example assumes that the state does not change from the first period to the second. As another instance, the decision made at the first period does not influence the information received at the second period. As yet another instance, the decision-maker chooses between the two sources of information at the ex-ante stage and does not get the opportunity to revise his choice later on.  We relax all these assumptions, and a few others,  in later sections. The analysis will turn out to be far less straightforward, but will share some similarities with the example presented here. It is because of these similarities that we start with the simplest (almost embarrassingly so) discounted problems before considering generalizations in later sections.

\section{Formulation of the Problem}\label{sec:model}

\subsection{The Problem} We start with an informal description of the problem. A decision-maker has to make repeated decisions in an uncertain (and possibly changing) environment. Prior to making a decision,  the decision-maker receives information, which may depend on the information received (and possibly the decisions made) in the past.   We want to compare two \emph{fixed} information structures (experiments), when the decision-maker values streams of payoffs according to the discounting criterion.\medskip

We start with the simplest possible formulation of the problem, analyze it, and then offer several generalizations in later sections. In the simplest formulation, the decision-maker receives information about an unknown \emph{fixed} state (or parameter) $\theta \in \Theta$ at periods $t=1,\dots, T \leq +\infty$. The information the decision-maker receives is independent of the decisions made. In what follows, with the exception of the number of periods, all sets are finite. For any collection of sets $S_1,\dots, S_{t}$, we write $S^t$ for $S_1 \times \dots \times S_{t}$, $s_t$ (resp., $s^t$) for a generic element of $S_t$ (resp., $S^t$).  We write random vectors in bold.

\medskip

\textbf{Experiments.} An experiment $f:=(f_1, f_2, \dots)$ is a  collection of kernels $f_t: \Theta \times X^{t-1}  \rightarrow \Delta(X_t)$, where $X_t$ is the set of signals at period $t$ and $X^{t-1}$  the set of past signals.  We write $f_t(x_t|x^{t-1},\theta)$ for the probability of the signal $x_t$ given $(x^{t-1},\theta)$ and $f^t(x^t|\theta)$ for $f_1(x_1|\theta) \times \dots \times f_t(x_t|x^{t-1},\theta)$. The experiment $f$ models the information the decision-maker receives over time: At each period $t$, the decision-maker observes the signal $x_t$. The objective is to compare the experiments $f$ and $g$, where the experiment $g$ is the collection of kernels $g_t: \Theta \times Y^{t-1}  \rightarrow \Delta(Y_t)$.

\medskip 

\textbf{Decision problems.} A decision problem is a tuple $(A_t,u_t)_t$, where $A_t$ is the set of decisions at period $t$, and $u_t: A_t \times \Theta \rightarrow \mathbb{R}$ the payoff function at period $t$.\medskip

\textbf{Discount factors.} A discount factor $\delta$ is a probability distribution over $\{1,\dots,T\}$, with $\delta_t$ the probability of $t$. Examples include geometric discounting ($\delta_t=\frac{1-\delta}{(1-\delta)^T}\delta^t$ for some $\delta \in (0,1)$),  quasi-hyperbolic discounting, averaging over finitely many periods ($\delta_t=1/T^*$ for some finite $T^*$), discounting of a single period ($\delta_t=1$ for some $t$), among others.\medskip

\textbf{Strategies.} Fix the decision problem $(A_t,u_t)_t$. We write $\sigma$ (resp., $\tau$) for the strategy adapted to the experiment $f$ (resp., $g$), that is, $\sigma$ is
the collection of maps $(\sigma_t:   A^{t-1} \times X^{t}  \rightarrow \Delta(A_t))_t$ (resp., $(\tau_t:  A^{t-1} \times Y^{t} \rightarrow \Delta(A_t))_t$). We denote $\proba_{\theta,f, \sigma}$ the probability over histories of signals and actions, induced by the strategy $\sigma$ and the experiment $f$, when the state is $\theta$, and write  $\mathbb{E}_{\theta,f, \sigma}$ for the expectation with respect to that probability. We define $\proba_{\theta,g, \tau}$ and  $\mathbb{E}_{\theta, g, \tau}$ similarly. \medskip

\textbf{Discounted problems.} A discounted problem  is such that $(A_t,u_t)=(A,u)$ for all $t$, and the decision-maker values the sequences of actions $(a_1,\dots,a_T)$ as the discounted sum $\sum_{t}\delta_t u(a_t,\theta)$. \medskip  

Before moving further, let us highlight the most important assumptions the above formulation imposes.  To start with, the state $\theta$ is fixed and the information the decision-maker receives over time does not depend on past and current decisions. As we shall see later, it is relatively straightforward to relax the former assumption, while it is harder to relax the latter.  The difficulty stems from the fact that we want to rank the two experiments, which are fixed, by  comparing the payoffs the decision-maker can obtain in \emph{all} decision problems.  In the simplest formulation,  changing the decision problem does not change the experiments to be compared. This is not so when past and current decisions influence the information the decision-maker receives -- this is the main challenge we will need to tackle. Also, the discount factor is fixed. In applications, it is sometimes convenient to consider a range of discount factors, e.g., all sufficiently high discount factors.  Lastly, the set of feasible decisions is the same at each period, hence independent of past decisions. Again, in applications, past decisions may constrain future decisions, e.g., past investments may constrain future investments. We present these generalizations and a few others  in later sections. While the above formulation is perhaps too simple, we decided to present it first as it helps to understand the more advanced results.  We now offer a brief review of the seminal results on the comparison of experiments in static problems.

\subsection{Comparing Experiments in Static Problems: A Brief Review} This section  reviews some of the main results on the comparison of experiments in static problems ($T=1$). We refer the reader to  \citet[Chapter 12]{blackwell-girshick}, \cite{torgersen91}, \cite{lecam96} and    \citet[Chapter 9]{strasser2011} for extensive reviews. \medskip 

The comparison of the experiments $f$ and $g$ rests on Definition \ref{def:valuable-static}, which states that
 experiment $f$ is more valuable than experiment $g$ if, regardless of the decision problem and the state, the decision-maker is always better-off with experiment $f$ than with experiment $g$. This definition, albeit in a slightly different form, first appeared in \cite{blackwell51}.\footnote{Blackwell's original definition compares  sets of attainable payoff profiles  $\bigcup_{\sigma} \{u \in \mathbb{R}^{|\Theta|}: u_{\theta} = \mathbb{E}_{\theta, f,\sigma}u(\mathbf{a},\theta), \forall \theta\}$ and $\bigcup_{\tau} \{u \in \mathbb{R}^{|\Theta|}: u_{\theta} = \mathbb{E}_{\theta, g,\tau}u(\mathbf{a},\theta), \forall \theta\}$, and requires the inclusion of the latter in the former for all $(A,u)$. \cite{blackwell51} attributed the definition to Bohnenblust, Shapley and Sherman (unpublished).} 

\begin{definition}\label{def:valuable-static}
The experiment $f$ is more valuable than the experiment $g$ in static problems if for all decision problems $(A,u)$, for all strategies $\tau$, there exists a strategy $\sigma$ such that $\mathbb{E}_{\theta, g,\tau}u(\mathbf{a},\theta) \leq \mathbb{E}_{\theta, f,\sigma}u(\mathbf{a},\theta)$ for all $\theta$. 
\end{definition}

The following theorem summarizes the main results on the comparison of experiments.

\begin{theorem}[\cite{blackwell51, blackwell53}] \label{th:blackwell}The following statements are equivalent: 
\begin{enumerate}[label=(\arabic*)]
\item  The experiment $f$ is more valuable than the experiment $g$ in static problems.
\item For all decision problems $(A,u)$, for all strategy $\tau$, there exists strategy $\sigma$ such that $\marg_{A}\proba_{\theta,g,\tau}=  \marg_{A} \proba_{\theta,f, \sigma}$, for all $\theta$.
\item Sufficiency: There exists $\gamma: X \rightarrow \Delta(Y)$ such that 
\begin{align*}
g(y|\theta)= \sum_{x}\gamma(y|x)f(x|\theta), \text{\;for all\;} (y,\theta). 
\end{align*}
\end{enumerate}
\end{theorem}

The equivalence between (1) and (2) first appeared in \citet{blackwell51} as Theorem 2(i), and the equivalence between (1)  and (3) as Theorems 5 and 6 in \cite{blackwell53} -- Blackwell attributed Theorem 5 to \cite{sherman51} and \cite{stein51}. See \cite{oliveira18}  for a recent proof. \medskip

A few remarks are worth making.  First, we compare the experiments state-by-state; an alternative is to assume a prior distribution and to compare the resulting expected payoffs. For all priors fully supported on $\Theta$, the two approaches are equivalent -- see \citet{blackwell51}.  Second, note that sufficiency is equivalent to the existence of a joint distribution of $(\mathbf{x},\mathbf{y})$ for which $\mathbf{x}$ is a sufficient statistics for $\theta$ (based on $(\mathbf{x},\mathbf{y})$).\footnote{For all $(x,y)$, let $h (x,y|\theta)$ be the  probability of $(x,y)$ when the state is $\theta$ and assume that $g(y|\theta)= \sum_x h(x,y|\theta)$ and $f(x|\theta)=\sum_y h(x,y|\theta)$. If $\mathbf{x}$ is sufficient for $\theta$ based on $(\mathbf{x},\mathbf{y})$, the factorization theorem states the existence of positive functions $h_1$ and $h_2$ such that $h(x,y|\theta)=h_1(x,y)h_2(x|\theta)$. It follows that $g(y|\theta) = \sum_{x: \sum_{\tilde{y}} h_1(x,\tilde{y})>0}\frac{h_1(x,y)}{\sum_{\tilde{y}}h_1(x,\tilde{y})}f(x|\theta)$, i.e., $g$ is a garbling of $f$ with $\gamma(y|x):=h_1(x,y)/\sum_{\tilde{y}}h_1(x,\tilde{y})$ whenever $\sum_{\tilde{y}}h_1(x,\tilde{y})>0$ the garbling. (We use the fact that if $\sum_{\tilde{y}} h_1(x,\tilde{y})=0$, then $h_1(x,\tilde{y})=0$ for all $\tilde{y}$ since $h_1$ is positive.)} A classical result in statistics then states for any estimator based on $(\mathbf{x},\mathbf{y})$, there exists an estimator based on the sufficient statistics $\mathbf{x}$ with the same risk. See  \citet[Theorem 6.1, p. 33]{lehmann-98}. In our context, the estimator  based on $(\mathbf{x},\mathbf{y})$ is $\tau$. The decision-maker can emulate $\marg_{A}\proba_{\theta,g,\tau}$ by first sampling a \emph{fictitious} $y$ with probability $\gamma(y|x)$ and then choosing $a$ with probability $\tau(a|y)$.  As we will see, the simulation of fictitious profiles of signals will continue to play a central role in our analysis. Third, another statement is in terms of the distribution of posterior beliefs the two experiments induce. The experiment $f$ is more valuable than the experiment $g$ if, and only if, the distribution of posteriors $f$ induces dominates the distribution of posteriors $g$ induces in the convex order. Some of our results can also be expressed in terms of posterior beliefs.

\section{Comparison of Experiments in Discounted Problems} \label{sec:main}

\subsection{Discounted Problems: The Basics}\label{sec:basics}

Comparing experiments in discounted problems necessitates to take a stance on several issues such as the point(s) in time at which comparisons are made or whether the comparison must hold for a single discount factor or a range of them. We start with the simplest definition, where the comparison is made at the ex-ante stage and holds for a fixed discount factor -- we consider other definitions later on.

\begin{definition}\label{def:valuable-discounted} 
The experiment $f$ is more valuable than the experiment $g$ in discounted problems with discount factor $\delta$ if for all decision problems $(A,u)$, for all  
strategy $\tau$, there exists strategy $\sigma$ such that 
\begin{align*}
\mathbb{E}_{\theta, g, \tau}\left[\sum_{t}\delta_t u(\mathbf{a}_t,\theta)\right] \leq \mathbb{E}_{\theta,f,\sigma}\left[\sum_{t}\delta_t u(\mathbf{a}_t,\theta)\right], \,\text{for all\;} \theta.
\end{align*} 
\end{definition}
\medskip

We now show how to derive a characterization similar to the one of \cite{blackwell51, blackwell53}, but for discounted problems. The main idea consists in reformulating discounted problems as static problems and then to apply Theorem \ref{th:blackwell}.  More specifically, observe that we can rewrite the decision-maker's expected payoff as follows:
\begin{align*}
\mathbb{E}_{\theta, g, \tau}\left[\sum_{t}\delta_t u(\mathbf{a}_t,\theta)\right]  & = \sum_{a} u(a,\theta)\underbrace{\left[\sum_t \delta_t \marg_{A_t}\proba_{\theta,g,\tau}(a)\right]}_{\text{discounted probability of $a$}}.
\end{align*}
In words, the decision-maker's expected payoff is the expectation of the static payoff with respect to the discounted probabilities of making decisions. 
As we shall see,  not all discounted problems admit such a reformulation, but our simplest formulation does. It immediately follows from Theorem \ref{th:blackwell} that $f$ is more valuable than $g$ if for all strategy $\tau$, there exists a strategy $\sigma$, which replicates the discounted probabilities of choices, that is, 
\begin{align*}
\sum_{t} \delta_t \marg_{A_t}\proba_{\theta,g,\tau} = \sum_{t} \delta_t \marg_{A_t}\proba_{\theta,f,\sigma}, \text{\;for all\;} \theta.
\end{align*}

Next, observe that we can rewrite the discounted probability of decision $a$ as follows:
\begin{align*}
\sum_t \delta_t \proba_{\theta,g,\tau}(\boldsymbol{a}_t=a) & = \sum_{t,y^t} \delta_t \proba_{\theta,g,\tau}(y^t)\proba_{\theta,g,\tau}(\boldsymbol{a}_t=a|y^t), \\
& = \sum_{t,y^t} \underbrace{\delta_t g^t(y^t|\theta)}_{\text{signal}} \underbrace{\proba_{\theta,g,\tau}(\boldsymbol{a}_t=a|y^t)}_{\text{indep. of $\theta \Rightarrow$ strategy}}.
\end{align*}
To understand the above rewriting, recall that the evolution of the signals does not depend on the decisions made and, consequently, the probability $\proba_{\theta,g,\tau}(y^t)$ is independent of the strategy $\tau$, and equals to $g^t(y^t|\theta)$. In addition, since the decision-maker does not observe $\theta$, the conditional probability $\proba_{\theta,g,\tau}(\boldsymbol{a}_t=a|y^t)$ is independent of $\theta$ and, thus, we can view it as a behavioral strategy mapping sequences of signals to probabilities over actions. In other words, the static re-formulation  states that the experiment $g$ generates the signal ``$(t,y^t)$'' with probability $\delta_t g^t(y^t|\theta)$ when the state is $\theta$, the decision-maker observes the realized signals, and makes a decision. (Recall that $y^t$ is an element of $Y^t$, the set of sequences of signals of length $t$.  Thus, summing over $(t,y^t)$ is equivalent to summing over $\cup_{t=1}^T Y^t$.)

Since the same re-formulation applies to the experiment $f$, Theorem \ref{th:blackwell} suggests that for $f$ to be more valuable than $g$,  we need the existence of garblings $\gamma_{t'}: X^{t'} \rightarrow \Delta(\cup_{t=1}^{T} Y^t)$ such that 
\begin{align*}
\delta_t g^t(y^t|\theta) = \sum_{t',x^{t'}} \delta_{t'} f^{t'}(x^{t'}|\theta) \gamma_{t'}(y^t|x^{t'}), \text{for all\;} (t,y^t,\theta).  
\end{align*}
The following theorem states that this is indeed the case. 

\begin{theorem}\label{th:fixed-discounting-uncontrolled} The following statements are equivalent: 
\begin{enumerate}[label=(\arabic*)]
\item The experiment $f$ is more valuable than the experiment $g$ in discounted problems with discount factor $\delta$.
\item For all decision problem $(A,u)$, for all  
strategy $\tau$, there exists a strategy $\sigma$ such that 
\begin{align*}
\sum_{t}\delta_t\marg_{A_t}\mathbb{P}_ {\theta, g, \tau} = \sum_{t} \delta_t\marg_{A_t} \mathbb{P}_{\theta, f,\sigma}, \text{for all\;} \theta.
\end{align*}
\item $\delta$-sufficiency: For all $t'=1,\dots,T$, there exists a garbling $\gamma_{t'}: X^{t'} \rightarrow \Delta(\cup_{t=1}^{T} Y^t)$ such that 
\begin{align*}
\delta_t g^t(y^t|\theta) = \sum_{t'=1}^{T}\sum_{x^{t'} \in X^{t'}} \delta_{t'} f^{t'}(x^{t'}|\theta) \gamma_{t'}(y^t|x^{t'}), \text{for all\;} (t,y^t,\theta).
\end{align*}
\end{enumerate}
\end{theorem}
Intuitively, $\delta$-sufficiency states that the information needed to emulate the experiment $g$ is contained in the experiment $f$, but $f$ does not  need to be more informative than $g$ at all periods, i.e., we do not need $g^t$ to be a garbling of $f^t$  at all $t$.\footnote{We can equivalently define a unique garbling $\gamma: \cup_{t}X^t \rightarrow \Delta(\cup_t Y^t)$.} For instance, $f$ can be less informative than $g$ at earlier periods if it is sufficiently more informative at later periods.  To put it simply, all we need is that there is more information in $f$ than $g$ on average, that is, we need the mixture experiment $\sum_t \delta_t f^t$ to be (Blackwell) sufficient for the mixture experiment $\sum_t \delta_t g^t$. To illustrate Theorem \ref{th:fixed-discounting-uncontrolled}, we revisit the introductory example. \medskip

\textbf{Introductory Example Revisited.} Recall that the introductory example compares two experiments $f$ and $g$, where $g$ provides some information at the first period and none at the second (i.e., $Y_2$ is the singleton $\{y_2\}$), while $f$ provides some information at the second period, but none at the first (i.e., $X_1$ is the singleton $\{x_1\}$).  In other words, the decision-maker trades off the benefit of receiving some information today with the benefit of receiving more information tomorrow. From Theorem \ref{th:fixed-discounting-uncontrolled}, $f$ is more valuable than $g$ if, only if, $f$ is $\delta$-sufficient for $g$, that is, if there exist garblings $(\gamma_1,\gamma_2)$ such that:
\begin{align*}
\delta_1 g^1(y_1|\theta) & = \delta_1 \gamma_1(y_1|x_1)f^1(x_1|\theta) + \delta_2 \sum_{x_2} \gamma_2(y_1|x_1,x_2)f^2(x_1,x_2|\theta), \\ 
\delta_2 g^2(y_1,y_2|\theta) & = \delta_1 \gamma_1(y_1,y_2|x_1)f^1(x_1|\theta) + \delta_2 \sum_{x_2} \gamma_2(y_1,y_2|x_1,x_2)f^2(x_1,x_2|\theta),
\end{align*}
for all $(y_1,y_2, \theta)$. (Recall that $\gamma_1: X_1 \rightarrow \Delta(Y_1 \cup (Y_1 \times Y_2))$ and $\gamma_2: X_1 \times X_2 \rightarrow \Delta(Y_1 \cup (Y_1 \times Y_2))$.)  Since $X_1$ and $Y_2$ are singletons, this is equivalent to finding two  garblings $\gamma_1' \in \Delta(Y_1)$ and $\gamma_2': X_2 \rightarrow \Delta(Y_1)$ such that 
\begin{align*}
g_1(y_1|\theta) & = \delta_1 \gamma'_1(y_1) + \delta_2 \sum_{x_2} \gamma'_2(y_1|x_2)f_2(x_2|\theta), \\
& = \sum_{x_2}\underbrace{\left[\delta_1 \gamma_1'(y_1) + \delta_2\gamma'_2(y_1|x_2)\right]}_{=:\widehat{\gamma}(y_1|x_2)}f_2(x_2|\theta),
\end{align*}
for all $(y_1,\theta)$.\footnote{If we have such kernels, simply define $\gamma_t(y_1|x^t)  = \delta_1 \gamma'_t(y_1|x^t)$ and $\gamma_t(y_1,y_2|x^t)=  \delta_2 \gamma'_t(y_1|x^t)$, assuming that $\delta_t>0$.} These two equations have a very natural interpretation. The second equation simply states that $f_2$ must been more informative than $g_1$ to compensate for the fact that the decision-maker receives the information at the latest period. The first equation simply states that the modified experiment $\widehat{f}$ must be more informative than the modified $\widehat{g}$ , where $\widehat{f}$ provides no information with probability $\delta_1$ and the information of $f$ with probability $\delta_2$, and $\widehat{g}$ is the same as $g$. Heuristically, we can represent a modified experiment as
\begin{table}[h]
\centering
\begin{tabular}{|c|c|c|c|c|} \hline
\backslashbox{$\Theta$}{$X$} & $\emptyset$ & $x_2$ & $x_2'$ & $\dots$ \\ \hline
$\theta$ & $\delta_1$ & $\delta_2f_2(x_2|\theta)$ & $\delta_2 f_2(x_2'|\theta)$ & $\dots$\\ \hline
$\theta'$ & $\delta_1$ &  $\delta_2f_2(x_2|\theta')$ & $\delta_2 f_2(x_2'|\theta')$ & $\dots$\\ \hline
$\vdots$ & $\vdots$ & $\vdots$ & $\vdots$ & $\vdots$\\ \hline
\end{tabular}\caption{The modified experiment $\widehat{f}$.}
\end{table}

It is then immediate to verify that the numerical example we presented earlier induces the modified experiment $\widehat{f}$ in Table \ref{Table:intro-modified-expe}. To check that $\widehat{f}$ is more valuable than $\widehat{g}$, we compare the two distributions over posteriors (assuming a uniform prior) these experiments induce. The distributions are in Table \ref{tab:posteriors}:
\begin{table}[h]
\centering
\begin{tabular}{|c|c|c|c|c|c|}\hline
posteriors of $\theta_0$ & $0$ & $5/12$ & $1/2$ & $7/12$ & $1$ \\ \hline
$\widehat{g}$ & $0$ & $1/2$ & $0$ & $1/2$ & $0$ \\ \hline
$\widehat{f}$ & $1/4$ & $0$ &$1/2$ & $0$ & $1/4$ \\ \hline
\end{tabular}
\caption{Distribution of posteriors}\label{tab:posteriors}
\end{table}
\medskip 

Clearly, the experiment $\widehat{f}$ induces a distribution, which is a mean-preserving spread of the distribution induced by $\widehat{g}$, hence $f$ is more valuable than $g$.\medskip 

To conclude that section, we remark that it is possible to make yet another equivalent statement in terms of posteriors, that is, the distribution over the posteriors induced by $\sum_{t} \delta_t f^t$ dominates, in the convex order, the distribution over the posteriors induced by $\sum_{t} \delta_t g^t$. See \cite{whitmeyer-williams-2024}  for a formal statement.

\subsection{A First Application: Comparing Bernouilli Experiments} 
We present a simple application, comparing Bernouilli experiments. Readers more interested in generalizations can skip this subsection.\medskip  

There are two periods. Let $\Theta=X_1=X_2=Y_1=Y_2=\{0,1\}$. The  experiment $f$ is defined by $f_1(0|0)=f_1(1|1) =p $ and $f_2(0|x_1,0)=f_2(1|x_1,1)=q$ for all $x_1 \in \{0,1\}$, with $p \geq q \geq 1/2$. Similarly, the experiment $g$ is defined by  $g_1(0|0)=g_1(1|1) =p' $ and $g_2(0|y_1,0)=g_2(1|y_1,1)=q'$ for all $y_1 \in \{0,1\}$, with $p' \geq q' \geq 1/2$. In words, the two experiments consist of two independent, but not necessarily identical, draws from two Bernouilli distributions, with the second-period distribution less informative than the first-period's one. \medskip 

Theorem \ref{th:fixed-discounting-uncontrolled} states that $f$ is $\delta$-sufficient for $g$ if, and only if, the mixture experiment $\sum_t \delta_t f^t$ is (Blackwell) sufficient for $\sum_t \delta_t g^t$. To the best of our knowledge, very little is known about the comparison of mixtures of experiments. \cite{torgersen1970} offers a brief discussion, but does not go much further than stating that if $f^t$ is sufficient for $g^t$ for all $t$, so is the mixture $\sum_t \delta_t f^t$ for $\sum_t \delta_t g^t$. We believe the result below is new.\medskip  

We now characterize when the mixture experiment $\delta_1 f^1 + \delta_2 f^2$ is sufficient for the mixture experiment $\delta_1 g^1 + \delta_2 g^2$, where $f^2$ is the  experiment given by $f^2(1,1|1)=f^2(0,0|0)=pq$, $f^2(1,0|0)= f^2(0,1|1)=(1-p)q$,    
$f^2(0,1|0)= f^2(1,0|1)=p(1-q)$. Similarly, for $g^2$. Since the mixture experiments are dichotomies, we can apply Theorem 12.4.1 of \cite{blackwell-girshick} to compare the two. The comparison requires to prove that the distribution over posteriors  induced by $\delta_1 f^1 + \delta_2 f^2$ is a \emph{mean-preserving} spread of the one induced by $\delta_1 g^1 + \delta_2 g^2$.\medskip

Table \ref{tab:posterior} summarizes the probabilities over signal realizations and posteriors induced by the experiment $\delta_1 f^1 + \delta_2 f^2$.\footnote{We assume a uniform prior. Recall, however, that the comparison is independent of the choice of the prior as long as it has full support.}

\begin{table}[h!]
\resizebox{\columnwidth}{!}{%
\begin{tabular}{|c|c|c|c|c|c|c|}\hline 
signal: & $0$ & $1$ & $0,0$ & $0,1$ & $1,0$ & $1,1$ \\ \hline
$[\theta=0]$ & $\delta_1 p$ & $\delta_1(1-p)$ & $\delta_2p q$ & $\delta_2p(1-q)$ & $\delta_2(1-p)q$ & $\delta_2(1-p)(1-q)$  \\ \hline
$[\theta=1]$ & $\delta_1(1-p)$ & $\delta_1p$ &  $\delta_2(1-p)(1-q)$ & $\delta_2(1-p)q$ & $\delta_2 p(1-q)$ & $\delta_2 p q$ \\ \hline 
$2 \times $ proba. signal: & $\delta_1$ & $\delta_1$ & $\delta_2(p q + (1-p)(1-q))$ & $\delta_2(p(1-q) + (1-p)q)$ & $\delta_2((1-p)q+p(1-q))$ & $\delta_2 ((1-p)(1-q) + pq)$ \\ \hline
posterior of $[\theta=0]$ & $p =: \pi_{0} $ & $1-p =: \pi_{1}$ & $\frac{pq}{p q + (1-p)(1-q)}=: \pi_{00} $ & $\frac{p(1-q)}{p(1-q) + (1-p)q}=: \pi_{01}$ & $\frac{(1-p)q}{p(1-q) + (1-p)q}=:\pi_{10}$   & $\frac{(1-p)(1-q)}{p q + (1-p)(1-q)}=:\pi_{11}$ \\ \hline
\end{tabular}}
\caption{Probability of signals and posteriors for the experiment $\delta_1 f^1 + \delta_2 f^2$}\label{tab:posterior}
\end{table} 
\medskip 

It is immediate to verify that $\pi_{11} \leq \pi_{1} \leq \pi_{10} \leq  1/2 \leq \pi_{01} \leq \pi_{0} \leq \pi_{00}$. The ranking between $\pi_{10}$ and $\pi_{01}$ follows from the observation that the first-period experiment is (weakly) more informative than the second-period experiment, so that observing first the signal ``0'' and then the signal ``1''  is more indicative of the state being $[\theta=0]$ than the reverse ordering.  From the symmetry of the experiments, we also have that $\pi_{00}=1-\pi_{11}$, $\pi_{0} =1-\pi_{1}$, and $\pi_{01} = 1-\pi_{10}$.  We write $\lambda:= pq + (1-p)(1-q)$ for the probability of two identical signals. By symmetry, $\lambda/2$ is then the probability of either the signal ``$(0,0)$'' or ``$(1,1)$.''  Note that $\lambda \pi_{11} + (1-\lambda)\pi_{10}=\pi_1$. The probabilities over signal realizations and posteriors induced by the experiment $g$ are obtained similarly. We indicate them with a prime throughout.

\begin{theorem}\label{th:comparison-bernouilli}
The experiment $f$ is $\delta$-sufficient for $g$ if, and only if, the following two conditions hold: (a) $\pi_1 \leq \pi'_1$ and (b): if $\delta_2>0$, one of the following inequalities must be satisfied: 
\begin{enumerate}
\item[(i)] The experiment $f^1$ is sufficient for the experiment $g^2$: $\pi_1\leq  \pi'_{11}$. 
 \item[(ii)]  The experiment $f^2$ is sufficient for the experiment $g^2$, but $f^1$ is not sufficient for $g^2$: 
 \begin{align*}
 \pi_{11} \leq \pi'_{11} < \pi_1 \text{\;and\;} 
(\pi_{10}-\pi_{11})\lambda \geq (\pi_{10}-\pi_{11}')\lambda'.
 \end{align*} 
 \item[(iii)]  Neither the experiment $f^1$ nor the experiment $f^2$ are sufficient for $g^2$: 
 \begin{align*}
 \pi_{11}& \leq \pi'_{11} < \pi_1, \\ 
 (\pi_{10}-\pi_{11})\lambda & < (\pi_{10}-\pi_{11}')\lambda', \\ 
 (\pi_1-\pi_{11}) \lambda  &\geq (\pi_{1}-\pi_{11}') \lambda',\\
(\pi_{10}- \pi_{11})\delta_2 \lambda & \geq 
(\pi_{10}-\pi_{11}')\delta_2 \lambda' + (\pi_1-\pi_1')\delta_1.
 \end{align*}
\end{enumerate}
\end{theorem}
Theorem \ref{th:comparison-bernouilli} provides a complete characterization for the $\delta$-sufficiency of $f$. First, it states that $f^1$ must be sufficient for $g^1$, i.e., $\pi_1 \leq \pi_1'$ (equivalently, $p \geq p'$). This is obviously necessary and sufficient when $\delta_1=1$. Perhaps surprisingly, this is also necessary when $\delta_2>0$. The reason is that the experiment $f^2$ is sufficient for $g^2$ only if $p \geq p'$.\footnote{More generally, when comparing two experiments consisting of two independent Bernouilli distributions each, a necessary condition to rank them is $\max(p,q) \geq \max(p',q')$.} Thus, if $p < p'$, then $g^1$ is sufficient for $f^1$ and, simultaneously, $f^2$ cannot be sufficient for $g^2$, so that either the two mixtures are not comparable or $g$ is $\delta$-sufficient.

Second, the conditions (i)-(iii) are mutually exclusive. Condition (i) states that $f^1$ is sufficient for $g^2$, that is, even receiving a single signal from $f$ is more informative than two signals from $g$. If it is the case, then $f^2$ is also sufficient for $g^2$ (since $f^2$ is necessarily more informative than $f^1$) and $f^1$ sufficient for $g^1$ (since $g^1$ is necessarily less informative than $g^2$).

Third, condition (ii) states that $f^2$ is more informative than $g^2$, but $f^1$ is not. Since we already have that $f^1$ is more informative than $g^1$, it follows that period-by-period, $f$ is more informative than $g$, which suffices for $\delta$-sufficiency. It is worth noting that if $p \geq p'$ and $q \geq q'$, then either condition (i) or (ii)  holds. None of the aforementioned conditions depend on the discount factors. 

We now examine more closely the condition $(\pi_{10}-\pi_{11})\lambda \geq (\pi_{10}-\pi_{11}')\lambda'$. We first note that the condition is equivalent to $(\pi_{10}-\pi_1) \geq (\pi_{10}-\pi_{10}')\lambda' + (\pi'_{10} - \pi_1')$ and, therefore, is satisfied when $\pi_{10}' \leq \pi_{10}$.\footnote{Recall that 
$\lambda \pi_{11} + (1-\lambda)\pi_{10}=\pi_1$, $\lambda' \pi'_{11} + (1-\lambda')\pi'_{10}=\pi'_1$, $ \pi_{11} \leq \pi'_{11}$, and $\pi_1 \leq \pi_1'$.}  
In this case, the two posteriors $\pi_{11}$ and $\pi_{10}$ are more dispersed than the  posteriors $\pi'_{11}$ and $\pi'_{10}$, which suffices by symmetry of our problem.  The real bite of the condition is when $\pi_{11} \leq \pi_{11} ' \leq \pi_{10} < \pi_{10}'$.  In this instance, the condition essentially says that conditional on the signals $(1,1)$ and $(1,0)$, the average of the posteriors $\pi_{11}$ and $\pi_{10}$ is less than the average of  $\pi_{11}'$ and $\pi_{10}'$. By symmetry, the average of the posteriors $\pi_{00}$ and $\pi_{01}$ is more than the average of  $\pi_{00}'$ and $\pi_{01}'$. That is, $f^2$ induces more dispersed posteriors on average.\footnote{The characterization is reminiscent  of \cite{wu2023geometric} geometric characterization of the Blackwell's order. The main difference is that the posteriors are not affinely independent in our case.}

Finally, condition (iii) is the only condition, which depends on the discount factor. In particular, it implies that  $f$ is $\delta$-sufficient for $g$ only if $\delta_2$ is not too high. Indeed, since we exclude the case where $f^2$ is sufficient for $g^2$, we cannot have $\delta_2$ too high. (If $\delta_2$ were too high, either the two experiments wouldn't be comparable or the reverse ordering would occur.)

\subsection{Controlling Information}\label{sec:control-info} Up to now, the analysis rests on the assumption that the information the decision-maker receives over time does not depend on the decisions made. While it may be natural in few applications, it is less so in many. For instance, in  dynamic pricing problems, the prices a seller offers influence not only the realized profits, but also what is learned about the demand. More generally, the assumption precludes the decision-maker from learning about realized payoffs over time -- a strong assumption.\footnote{As signals, the realized payoffs depend on the decisions made.}  \medskip

The most natural idea is to incorporate the decisions into the statistical experiments, i.e., to write kernels as  $f_t: \Theta  \times X^{t-1} \times A^{t-1} \rightarrow \Delta(X_t)$. There is, however, a major problem with this modeling: the statistical experiments change as we change the decision problems. To extend the approach of \cite{blackwell51, blackwell53}, we must be able to vary the decision problems without changing the experiments to be compared.\medskip

The solution we adopt parallels the classical distinction between decisions/actions and consequences. For instance, in mechanism design theory, there is a set of messages (bids, votes, etc) players can choose from and a  map from messages to consequences (allocations, elected candidates, etc) -- messages are not consequences. A similar distinction is made in decision theory. We  follow a similar approach in that we distinguish between decisions and their \emph{informational} consequences. More precisely, we introduce sets of (informational) controls $K_t$ and maps $\kappa_t$ from actions to controls. The sets of controls are fixed and enter into the definition of statistical experiments, while the maps from actions to controls are part of the decision problem.  We view the controls as the subset of consequences, which directly influence the information the decision-maker receives over time. For instance, in a simple dynamic pricing problem, a seller chooses prices (the decisions), which determine the realized demands, points on the demand curve (the controls) the seller is learning about. In a more realistic pricing problem, a seller chooses not only prices, but also marketing strategies.  Different strategies, e.g., a mass campaign in national newspapers or targeted campaigns through social media, reach consumers differently, thus informing some consumers about price changes, while leaving others unaware. In such a problem, the seller only learns about the demand of consumers made aware of the prices offered. Prices and marketing strategies (the decisions) determine, possibly probabilistically, the sample of informed consumers and its demand (the controls), a noisy signal about the true demand (the state). In what follows, we give a formal description of the solution we adopt and derive a characterization for the comparison of statistical experiments. \medskip

\textbf{Experiments.} An experiment $f:=(f_1, f_2, \dots)$ is a  collection of kernels $f_t: \Theta \times K^{t-1} \times X^{t} \times \Theta \rightarrow \Delta(X_t)$, where $X_t$ is the set of signals at period $t$ and $K_t$ the set of controls. We write $f^t(x^t|k^{t-1},\theta)$ for $f_1(x_1|\theta) \times \dots \times f_t(x_t|x^{t-1},k^{t-1},\theta)$. Throughout, we compare the experiments $f$ and $g$, where the experiment $g$ is the collection of kernels $g_t: \Theta \times K^{t-1} \times Y^{t}  \rightarrow \Delta(Y_t)$.

\medskip 

\textbf{Decision problems.} A decision problem is a tuple $(A_t,u_t,\kappa_t)_t$, where $A_t$ is the set of decisions at period $t$, $u_t: A_t \times \Theta \rightarrow \mathbb{R}$ the payoff function at period $t$, and $\kappa_t:   A_t \times K^{t-1}\rightarrow \Delta(K_t)$ the control map. The control map summarizes how the decisions influence the information the decision-maker receives over time. Throughout, we write $\kappa$ for $(\kappa_t)_t$.\medskip  

\textbf{Strategies.} Fix the decision problem $(A_t,u_t, \kappa_t)_t$. We write $\sigma$ (resp., $\tau$) for the strategy adapted to the experiment $f$ (resp., $g$), that is, $\sigma$ is
the collection of maps $(\sigma_t:   A^{t-1} \times X^{t} \times K^{t-1}  \rightarrow \Delta(A_t))_t$ (resp., $(\tau_t:  A^{t-1} \times Y^{t} \times K^{t-1} \rightarrow \Delta(A_t))_t$). We denote $\proba_{\theta,f, \kappa, \sigma}$ the probability over histories of signals, controls and actions, induced by the strategy $\sigma$ and the experiment $f$, when the state is $\theta$, and write  $\mathbb{E}_{\theta,f, \kappa, \sigma}$ for the expectation with respect to that probability. We define $\proba_{\theta,g, \kappa, \tau}$ and  $\mathbb{E}_{\theta, g, \kappa, \tau}$ similarly. \medskip

We make three important assumptions. First, we assume that the decision-maker observes the realized controls. This assumption allows us to focus on beliefs about the state, which is our main interest. Without this assumption, beliefs about realized controls would also be instrumental.   Second, we assume that the control maps do not depend on the state. This assumption is needed to maintain a sharp dichotomy between statistical experiments and decision problems. If the control maps were to depend on the state, the realized controls would also be informative signals about the state -- an additional statistical experiment. Third, we assume that only the current choice affects the current control; past choices affect the current control only through the realization of past controls $K^{t-1}$. 

\medskip 

We now turn to the comparison between the statistical experiments $f$ and $g$. We continue to say that $f$ is more valuable than $g$ if for all decision problems $(A_t,u_t,\kappa_t)_t$, for all strategy $\tau$, there exists a strategy $\sigma$, which gives a (weakly) higher payoff in all states.  (Recall that $(A_t,u_t)=(A,u)$ for all $t$.) As in the uncontrolled case, it is immediate to see that $f$ is more valuable than $g$ if, and only if, for all discounted problems $(A,u,\kappa)$, for all strategy $\tau$, there exists a strategy $\sigma$ such that 
\begin{align*}
\sum_{t}\delta_t\marg_{A_t}\mathbb{P}_ {\theta, g,\kappa, \tau} = \sum_{t} \delta_t\marg_{A_t} \mathbb{P}_{\theta, f,\kappa,\sigma}.  
\end{align*}
Moreover, we can still write the discounted probability of $a$ as:
\begin{align*}
\sum_t \delta_t \proba_{\theta,g,\kappa, \tau}(\boldsymbol{a}_t=a) & = \sum_{t,y^t,k^{t-1}} \delta_t \proba_{\theta,g,\kappa, \tau}(y^t,k^{t-1})\proba_{\theta,g, \kappa,\tau}(\boldsymbol{a}_t=a|y^t,k^{t-1}).
\end{align*}
Unlike the uncontrolled case, however, it is no longer true that the probability $\proba_{\theta,g,\kappa, \tau}(y^t,k^{t-1})$ is independent of the strategy $\tau$, since the strategy influences the realization of the controls and, consequently, of the signals. Yet, we now argue that a similar idea continues to work.\medskip 

 We start with an informal exposition. Observe that the probability measure $\proba_{\theta,g,\kappa, \tau}$ induces kernels $\xi^{\boldsymbol{y}}_t: Y^{t} \times K^{t-1} \rightarrow \Delta(K_t)$, defined by $\xi^{\boldsymbol{y}}_t(k_t|y^t,k^{t-1}):= \proba_{\theta,g,\kappa, \tau}(k_t|y^t,k^{t-1})$  whenever $\proba_{\theta,g,\kappa, \tau}(y^t,k^{t-1}) >0$. In turn, this defines an \emph{extended} experiment $\proba_{\theta,g,\xi^{\boldsymbol{y}}}$ over profiles of signals and controls, where the probability of $(y^t,k^{t-1})$  is
 \[g_1(y_1|\theta)\times \xi_1^{\boldsymbol{y}}(k_1|y^1,k^0) \times \dots \times \xi_t^{\boldsymbol{y}}(k_t|y^{t},k^{t-1}) \times g_t(y_{t+1}|y^{t},k^{t},\theta),\]
 when the state is $\theta$. (We let $K_0$ be the singleton $\{k_0\}$.)  Analogously, the probability measure $\proba_{\theta,f,\kappa, \sigma}$ defines the kernels $\xi^{\boldsymbol{x}}_t: X^{t} \times K^{t-1} \rightarrow \Delta(K_t)$ and the extended experiment $\proba_{\theta,g,\xi^{\boldsymbol{x}}}$. The broad idea is to treat the extended experiments $\proba_{\theta,g,\xi^{\boldsymbol{y}}}$ and $\proba_{\theta,g,\xi^{\boldsymbol{x}}}$ as two \emph{uncontrolled} statistical experiments and to apply the result from the previous section.\medskip

 There are two important caveats, though. First, as the decision problem varies, so does $\kappa$ and, consequently, the kernels $ \xi^{\boldsymbol{y}}$ and  $\xi^{\boldsymbol{x}}$. To see this, simply note that $ \xi_1^{\boldsymbol{y}}(k_1|y^1,k^0) = \sum_{a_1}\tau_1(a_1|y_1)\kappa_1(k_1|a_1,k^0)$. The solution is to consider a collection of such kernels. Second, and most importantly, the kernels $ \xi^{\boldsymbol{y}}$ and  $\xi^{\boldsymbol{x}}$ are not independent. As already mentioned, they both depend on the same $\kappa$. There is another, more subtle, interdependence.  Recall that the aim is to emulate the discounted distribution over actions of any strategy $\tau$ (adapted to $g$)  by a strategy $\sigma$  (adapted to $f$), i.e., for all $\tau$, there must exist $\sigma$ such that
 \begin{align*}
\sum_{t}\delta_t\marg_{A_t}\mathbb{P}_ {\theta, g,\kappa, \tau} = \sum_{t} \delta_t\marg_{A_t} \mathbb{P}_{\theta, f,\kappa,\sigma}.  
\end{align*}
 So, suppose, as we do in static problems, that we emulate the distribution over the first-period actions by first sampling a fictitious $y_1$ with probability $\gamma_1(y_1|x_1)$ when the first-period signal is $x_1$ and then choosing $a_1$ with probability $\tau_1(a_1|y_1)$, that is, $\sigma_1$ is given by $\sigma_1(a_1|x_1):= \sum_{y_1}\tau_1(a_1|y_1)\gamma_1(y_1|x_1)$. It follows that  
 \begin{align*}
 \xi_1^{\boldsymbol{x}}(k_1|x^1,k^0)= \sum_{y_1} \underbrace{\left(\sum_{a_1}\kappa_1(k_1|a_1,k^0)\tau_1(a_1|y^1)\right)}_{=\xi_1^{\boldsymbol{y}}(k_1|y^1,k^0)} \gamma_1(y_1|x_1),
 \end{align*}
 that is, $ \xi_1^{\boldsymbol{x}}$ is a ``garbled'' version of $ \xi_1^{\boldsymbol{y}}$. Emulating $\tau$ with the strategy $\sigma$ forces the two kernels $\xi^{\boldsymbol{y}}$ and  $\xi^{\boldsymbol{x}}$ to be interdependent. For another perspective, suppose that the strategy $\tau$ always picks $a$, while $a$ is associated with relatively ``uninformative'' controls under $f$.\footnote{That is, for these controls $k^{t-1}$, the experiments $f_t(\cdot|\cdot,k^{t-1},\theta)$ are not too informative.} Then, playing $a$ under $\sigma$  limits the informativeness of future $X$-signals. In other words, using information from $f$ at early periods to simulate the information from $g$ constrains the information available at later dates. The characterization needs to capture this dynamic effect.\medskip

We turn to a formal exposition. Let $\xi:=(\xi_t)_t$ be a \emph{test kernel} with $\xi_t: \bigcup_{t'=1}^{T}(Y^{t'} \times K^{t'-1}) \times K^{t-1} \rightarrow \Delta(K_t)$ and $\Xi$ the set of all test kernels. We interpret the kernel $\xi_t$ as a map from profiles of \emph{simulated} signals and controls $(y^{t'},l^{t'-1})$ and \emph{realized} controls $k^{t-1}$ to probabilities over controls at period $t$. Given $\xi$, the  probability $\mathbb{P}_{\theta, g, \;\xi}$ over profiles of signals and controls (in $\bigcup_{t=1}^{T}( Y^t \times K^{t-1})$) is defined by: 
\begin{align*}
\mathbb{P}_{\theta, g, \;\xi}(y^t,l^{t-1}) = & g_1(y_1|\theta) \times \xi_1(l_1|(y^1,l^0),l^0) \times \dots \times \\
& g_{t-1}(y_{t-1}|y^{t-2},l^{t-2},\theta) \times \xi_{t-1}(l_{t-1}|(y^{t-1},l^{t-2}),l^{t-2}) \times \\
& g_{t}(y_t|y^{t-1},l^{t-1},\theta),
\end{align*}
for all $(y^t,l^{t-1}) \in Y^t \times K^{t-1}$.  The definition of a test kernel may at first be surprising -- a reader would perhaps expect it to be defined from $Y^t \times K^{t-1}$ to $\Delta(K_t)$, as we did with $\xi^{\boldsymbol{y}}$. The need to adopt a more complex definition follows from the dual role garblings  have, that is, garblings emulate strategies and, simultaneously, influence the flow of information through the controls. This parallels the dual role decisions have, as determining payoffs and, simultaneously, influencing the flow of information. Formally, as in the uncontrolled case, consider the profile of garblings $\gamma:=(\gamma_t)_t $, where $\gamma_t: X^t \times K^{t-1} \rightarrow \Delta\left( \bigcup_{t'=1}^{T}(Y^{t'} \times K^{t'-1})\right)$ for each $t=1,\dots,T$. Given $(\xi,\gamma)$, we define the  probability $\mathbb{P}_{\theta, f, \xi \circ \gamma}$ over signals and controls (in $\bigcup_{t=1}^{T} (X^{t} \times K^{t-1})$) by:
\begin{align*}
\mathbb{P}_{\theta, f, \xi \circ \gamma}(x^t,k^{t-1}) = & f_1(x_1|\theta) \times \left(\sum_{t',y^{t'},l^{t'-1}}\xi_1(k_1|(y^{t'},l^{t'-1}),k^0)\gamma_1(y^{t'},l^{t'-1}|x^1,k^0)\right) \times \dots  \times \\
 & f_{t-1}(x_{t-1}|x^{t-2},k^{t-2},\theta) \times \left(\sum_{t',y^{t'},l^{t'-1}}\xi_{t-1}(k_{t-1}|(y^{t'},l^{t'-1}),k^{t-2})\gamma_{t-1}(y^{t'},l^{t'-1}|x^{t-1},k^{t-2})\right) \times \\
 & f_t(x_t|x^{t-1},k^{t-1},\theta),
\end{align*}
for all $(x^t,k^{t-1})$. In words, given $(x^{t-1},k^{t-2})$, the decision-maker \emph{simulates} a fictitious history $(y^{t'},l^{t'-1})$ of signals and controls with probability $\gamma_{t-1}(y^{t'},l^{t'-1}|x^{t-1},k^{t-2})$, then draws the control $k_{t-1}$ with probability $\xi_{t-1}(k_{t-1}| (y^{t'},l^{t'-1}), k^{t-2})$, that is, according to the test function $\xi_{t-1}$, given the simulated history $(y^{t'},l^{t'-1})$ and the actual realization of past controls $k^{t-1}$, and then receives the signal $x_t$ with probability $f_t(x_t|x^{t-1},k^{t-1}, \theta)$, when the state is $\theta$.\footnote{We stress that $t'$ may differ from $t$: At period $t$, the decision-maker simulates histories of signals and controls $(y^{t'}, l^{t'-1})$ of any length $t'$.} The probability of transitioning from $(x^{t-1},k^{t-2})$ to $(x^{t-1},k^{t-1})$ is therefore
\begin{align*}
(\xi_{t-1} \circ \gamma_{t-1}) (k_{t-1}|x^{t-1},k^{t-2}):=\sum_{t',y^{t'},l^{t'-1}}\gamma_{t-1}(y^{t'},l^{t'-1}|x^{t-1},k^{t-2})\xi_{t-1}(k_{t-1}|(y^{t'},l^{t'-1}),k^{t-2}).
\end{align*}

For any collection of test kernels $\xi$ and garblings $\gamma$,  we thus generate two linked processes $\mathbb{P}_{\theta, g, \;\xi}$ and $\mathbb{P}_{\theta, f, \xi \circ \gamma}$ over signals and controls, which we then treat as \emph{uncontrolled} experiments. We then apply the condition of $\delta$-sufficiency on these two experiments, with the additional consistency requirement that the garblings defining  $\mathbb{P}_{\theta, f, \xi \circ \gamma}$ coincide with the garbling defining $\delta$-sufficiency.

\begin{theorem}\label{fixed-discounting} The following statements are equivalent: 
\begin{enumerate}[label=(\arabic*)]
\item The experiment $f$ is more valuable than the experiment $g$ in all discounted problems, where the decision-maker's discount factor is $\delta$.
\item For all decision problems $(A,u, \kappa)$, for all  
strategy $\tau$, there exists a strategy $\sigma$ such that 
\begin{align*}
\sum_{t}\delta_t\marg_{A_t}\mathbb{P}_ {\theta, g,\kappa, \tau} = \sum_{t} \delta_t\marg_{A_t} \mathbb{P}_{\theta, f,\kappa,\sigma},  \text{for all\;} \theta.
\end{align*}
\item $\delta$-sufficiency: For all $\xi \in \Xi$, for all $t'=1,\dots,T$, there exists a kernel $\gamma_{t'}: X^{t'} \times K^{t'-1} \rightarrow \Delta\left(\bigcup_{t=1}^{T}(Y^t \times K^{t-1})\right)$ such that 
\begin{align*}
\delta_t \mathbb{P}_{\theta,g,\,\xi}(y^t,l^{t-1}) = \sum_{t'=1}^{T} \sum_{x^{t'} \in X^{t'},  k^{t'-1} \in K^{t'-1}}\delta_{t'}\mathbb{P}_{\theta,f,\,\xi\circ \gamma}(x^{t'},k^{t'-1})\gamma_{t'}(y^t,l^{t-1}|x^{t'},k^{t'-1}),
\text{for all\;} (t,y^t,l^{t-1}, \theta).
\end{align*} 
\end{enumerate}
\end{theorem}

It is instructive to reformulate the third statement as follows: for all $\xi$, there exist garblings $\gamma$ and $\widehat{\gamma}$ such that
\begin{align}\label{eq:dual-role}
\delta_t \mathbb{P}_{\theta,g,\,\xi}(y^t,l^{t-1}) = \sum_{t', x^{t'}, k^{t'-1}}\delta_{t'}\mathbb{P}_{\theta,f,\,\xi\circ \widehat{\gamma}}(x^{t'},k^{t'-1})\gamma_{t'}(y^t,l^{t-1}|x^{t'},k^{t'-1}),
\text{for all\;} (t,y^t,l^{t-1}, \theta),
\end{align} 
and $\widehat{\gamma} =\gamma$.  Write $\mathbb{P}^t_{\theta,g,\,\xi}$ for the marginal of $\mathbb{P}_{\theta,g,\,\xi}$ over $Y^t \times K^{t-1}$ and, analogously, for $\mathbb{P}^t_{\theta,f,\,\xi\circ \widehat{\gamma}}$. The role of the garbling $\gamma$ is then to simulate the experiment $\sum_t \delta_t \mathbb{P}^t_{\theta,g,\,\xi}$ from the experiment $\sum_t \delta_t \mathbb{P}^t_{\theta,f,\,\xi\circ \widehat{\gamma}}$, as in \cite{blackwell51} and Section \ref{sec:basics}.\footnote{Note that if each $K_t$ is a singleton, then $\mathbb{P}^t_{\theta,g,\,\xi}$ (resp., $\mathbb{P}^t_{\theta,f,\,\xi\circ \widehat{\gamma}}$) is $g^t (\cdot|\theta)$ (resp.,) $f^t(\cdot|\theta)$.} We stress that the garbling $\gamma$ depends on both the signals and controls because not only signals, but also controls, are informative about the states. As already mentioned, however, using information from one controlled experiment $f$ to simulate another controlled experiment $g$ also has an endogenous effect on the informativeness of $f$. The role of the garbling $\widehat{\gamma}$ is to capture that effect. Finally, the consistency requirement $\widehat{\gamma} =\gamma$ stipulates that a unique garbling should play these two roles.\medskip

We conclude with a few additional remarks. First, it is easy to check that if all sets of controls are singletons, we recover the characterization we presented earlier in Theorem \ref{th:fixed-discounting-uncontrolled}. Second, the condition of $\delta$-sufficiency requires to consider \emph{all} test kernels, a consequence of considering all control maps $\kappa$. This is not uncommon in mathematics, e.g., the weak convergence of probability measures requires the converge of their expected values for \emph{all} continuous and bounded functions.  Yet, it would be nice to restrict attention to a smaller class of kernels, e.g., the deterministic kernels. An alternative would be to restrict the set of control maps, e.g., to associate ``higher'' controls with ``higher'' actions. We leave these open questions for future research. Third, an algorithm to check for $\delta$-sufficiency is as follows. For each $\widehat{\gamma}$,  associate all  garblings $\gamma$, if any, such that Equation (\ref{eq:dual-role}) holds, a linear programming problem. This defines a convex-valued correspondence from the set of garblings to its power set. If for all $\xi$, such a correspondence admits a fixed point, then $\delta$-sufficiency holds.  

\subsection{A Second Application: Controlling the Arrival Time of News}\label{sec:example-controlled}
We compare two \emph{controlled} experiments $f$ and $g$, which differ in their arrival time of the same information -- the control influences the arrival time.\footnote{\cite{whitmeyer-williams-2024} study the uncontrolled case.}  The decision-maker observes  $z \in Z$ at period $t \in \{1,2\}$ with probability $\alpha^k_t h(z|\theta)$ and $\beta^k_t h(z|\theta)$, respectively, when the parameter is $\theta \in \Theta$  and the control $k \in K$. 
The decision-maker observes nothing with the complementarity probabilities  $\alpha_{\emptyset}^k:=1-\alpha_1^k - \alpha_2^k$  and $\beta_{\emptyset}^k= 1-\beta_1^k - \beta_2^k$, respectively. We assume that $\alpha_1^k$ and $\beta_1^k$ are independent of $k$.\footnote{Formally, we assume that $K_0$ is a singleton and write $k$ for $k_1 \in K_1$. When $K_0$ is not a singleton, the same arguments apply $k_0$-by-$k_0$.} In words,  the decision-maker  controls the probability of observing the signal $z$ at the second period, conditional on  $z$ not being observed at the first period.\medskip  

With the same notation as in Section \ref{sec:control-info}, the experiment $f$ reads:  
\begin{align*}
& f_1(z|\theta)  = \alpha_1 h(z|\theta) , f_1(\emptyset|\theta)  = 1-\alpha_1, \\
& f_2(\emptyset|(k,z), \theta)  = 1, f_2(z|(k, \emptyset),\theta) = \frac{\alpha^k_2}{1-\alpha_1}h(z|\theta),  f_2(\emptyset|(k,\emptyset), \theta)  = \frac{\alpha_0^k}{1-\alpha_1},
\end{align*}
where we use the symbol $\emptyset$ to denote the absence of signals at a period. A similar description holds for $g$. 
\medskip

We now derive sufficient conditions for $f$ to be more valuable than $g$. Recall that for all $\xi$, we need to find a garbling $\gamma$ such that \[
\delta_{t'} \mathbb{P}_{\theta,g, \xi}(y^{t'}, l^{t'-1}) = \sum_{t,k^{t-1},x^{t}}\delta_t \gamma(y^{t'}, l^{t'-1}| x^t, k^{t-1}) \mathbb{P}_{\theta,f, \xi \circ \gamma}(k^{t-1},x^t),\] 
for all $(t', y^{t'}, l^{t'-1})$.\medskip 

Fix $\xi$. We first derive the expression $\mathbb{P}_{\theta,g, \xi}$ for all possible profiles of controls and observations, that is, $(z)$, $(\emptyset)$, $((z,\emptyset),k)$, $((\emptyset,z), k)$, and $((\emptyset,\emptyset), k)$, for all $(z,k)$. For instance,   $((z,\emptyset), k)$ means that the decision-maker observed $z$ in the first period and the realized control is $k$. As another instance, $((\emptyset,\emptyset), k)$ means that no observation has been made by the end of the second period, while the realized control is $k$. We have the following: 
\begin{align*}
\mathbb{P}_{\theta,g, \xi}(z)  & = \beta_1 h(z|\theta), \\
 \mathbb{P}_{\theta,g, \xi}(\emptyset) & = 1-\beta_1, \\
\mathbb{P}_{\theta,g, \xi}((z,\emptyset), k)  & = \beta_1 h(z|\theta)\xi(k|z), \\
 \mathbb{P}_{\theta,g, \xi}((\emptyset, z), k) &  = \beta^k_2 h(z|\theta)\xi(k|\emptyset),\\
\mathbb{P}_{\theta,g, \xi}((\emptyset, \emptyset), k) &  = \beta_{\emptyset}^k \xi(k|\emptyset).
\end{align*}
Similarly, $\mathbb{P}_{\theta,f, \xi \circ \gamma}$ takes the following values:
\begin{align*}
\mathbb{P}_{\theta,f, \xi \circ \gamma}(z)  & = \alpha_1 h(z|\theta), \\
\mathbb{P}_{\theta,f, \xi \circ \gamma}(\emptyset)  & = 1-\alpha_1, \\
\mathbb{P}_{\theta,f, \xi \circ \gamma}((z,\emptyset), k)  & = \alpha_1 h(z|\theta) \left(\sum_{\tilde{z}}\gamma(\tilde{z}|\hat{z})\xi(k|\tilde{z})\right), \\
 \mathbb{P}_{\theta,f, \xi \circ \gamma}((\emptyset, z), k) &  = \alpha_2^k h(z|\theta) \left(\sum_{\tilde{z}}\gamma(\tilde{z}|\emptyset)\xi(k|\tilde{z})\right),\\
\mathbb{P}_{\theta,f, \xi \circ \gamma}((\emptyset, \emptyset), k) &  = \alpha_{\emptyset}^k \left(\sum_{\tilde{z}}\gamma(\tilde{z}|\emptyset)\xi(k|\tilde{z})\right).
\end{align*}

Thus, for $f$ to be more valuable than $g$, we need to prove the existence of a garbling $\gamma$ such that:
\begin{align*}
\delta_{t} \mathbb{P}_{\theta,g, \xi}(y^{t}, l^{t-1})   = & \delta_1 \sum_{\hat{z}}\gamma(y^{t}, l^{t-1}|\hat{z})  \alpha_1 h(\hat{z}|\theta) + \\
&\delta_1 \gamma(y^{t}, l^{t-1}|\emptyset)  (1-\alpha_1) +\\
&\delta_2 \sum_{\hat{z},k}\gamma(y^{t}, l^{t-1}|(\hat{z},\emptyset), k) \alpha_1 h(\hat{z}|\theta) \left(\sum_{\tilde{z}}\gamma(\tilde{z}|\hat{z})\xi(k|\tilde{z})\right) + \\
&\delta_2 \sum_{\hat{z},k}\gamma(y^{t}, l^{t-1}| (\emptyset, \hat{z}), k)\alpha_2^k h(\hat{z}|\theta) \left(\sum_{\tilde{z}}\gamma(\tilde{z}|\emptyset)\xi(k|\tilde{z})\right) +\\
&\delta_2 \sum_{\hat{z},k}\gamma(y^{t}, l^{t-1}|(\emptyset, \emptyset), k) \alpha_{\emptyset}^k  \left(\sum_{\tilde{z}}\gamma(\tilde{z}|\emptyset)\xi(k|\tilde{z})\right),
\end{align*}
for all possible profiles of controls and observations $(t, y^{t}, l^{t-1})$. We have the following result.

\begin{proposition} \label{prop:ex-controlled}
The controlled experiment $f$ is $\delta$-sufficient for the controlled experiment $g$ if, and only if, $\alpha_1 +\delta_2 \alpha_2^k \geq \beta_1 + \delta_2 \beta_2^k$ for all $k$. 
\end{proposition}

To get some intuition, assume that $\xi$ puts probability one on the same fixed $k$ at all profiles of simulated controls and signals, and realized controls. (Recall that $\xi_t$ maps profiles of simulated controls and signals and realized controls $\bigcup_{t'=1}^{T}(Y^{t'} \times K^{t'-1}) \times K^{t-1}$ into $\Delta(K_t)$.)  So, in effect, we are comparing two \emph{uncontrolled} experiments, where the probabilities of arrival are $\alpha_t^{k}$, $t \in \{1,2,\emptyset\}$. Now, whether the observation $z$ arrives at the first or the second period, the posterior belief is the same. The probability of that posterior belief in the mixture $\sum_t \delta_t f^t$ (resp., $\sum_t \delta_t g^t$) is $(\alpha_1 + \delta_2 \alpha_2^k)h(z|\theta)$ (resp., $(\beta_1+ \delta_2\beta^k_2)h(z|\theta)$). With the complementary probability, no signal is observed, and the posterior belief is the prior. Thus, if $\alpha_1 +\delta_2 \alpha_2^k \geq \beta_1 + \delta_2 \beta_2^k$, the experiment $f$ is more informative than the experiment $g$, when the control is fixed at $k$. Since the choice of $k$ is arbitrary, this must hold for all $k$. Conversely, if $\alpha_1 +\delta_2 \alpha_2^k < \beta_1 + \delta_2 \beta_2^k$ for some $k$, then $f$ cannot be $\delta$-sufficient for $g$.\medskip 

Next, suppose that $\xi$ is not degenerated. Clearly, whenever the decision-maker observes $z$ at the first period, the realized control is irrelevant: the decision-maker already has all the information he can get. Therefore, the realized control only matters when there is no observation at the first period, i.e., after $(\emptyset)$.  If $\alpha_1 +\delta_2 \alpha_2^k \geq \beta_1 + \delta_2 \beta_2^k$ for all $k$, then $\alpha_1 +\delta_2 \sum_{k}\alpha_2^k\xi(k|\emptyset) \geq \beta_1 + \delta_2 \sum_{k}\beta_2^k \xi(k|\emptyset)$ for all $\xi$.  This last inequality allows us to extend the arguments for degenerated $\xi$ to arbitrary $\xi$. The proof is in Appendix \ref{app:ex-controlled}. \medskip 

To conclude, we derive ``simpler'' sufficient conditions to guarantee that $f$ is $\delta$-sufficient for $g$. To do so, note that we can rewrite the necessary and sufficient condition as
\[(\delta_1-\delta_2)[\beta_1] + \delta_2[\beta_1 + (\beta_1+\beta_2^k)] +  \leq  (\delta_1-\delta_2)[\alpha_1] + \delta_2[\alpha_1 + (\alpha_1+\alpha_2^k)].\]
Thus, if the discount factor is decreasing, it suffices that the cumulative distribution $(\beta_1, \beta_1+\beta_2^k, 1)$ of arrival times $(1,2,+\infty)$  \emph{second-order stochastically dominates} the cumulative distribution $(\alpha_1, \alpha_1+\alpha_2^k, 1)$, where an arrival time of $+\infty$ means that the information never arrives.\footnote{The sufficient condition in terms of second-order dominance and decreasing discount factors is the one developed in \cite{whitmeyer-williams-2024}.} In addition, if the first distribution \emph{first-order stochastically dominates} the second, we do not need decreasing discount factors.

\subsection{Discount Factors} \label{sec:discount factors} This section examines the role of the discount factor, which we have assumed fixed so far. We limit the discussion to the setting of subsection \ref{sec:basics}. We first argue that the condition of $\delta$-sufficiency satisfies convexity, that is, if experiment $f$ is $\delta$- and $\widehat{\delta}$-sufficient   for $g$, then it is also $\alpha \delta + (1-\alpha) \widehat{\delta}$-sufficient for $g$, for all $\alpha \in [0,1]$.

\begin{proposition} \label{prop:convexity} The set $\{\delta: f \text{\;is\;} \delta-\text{sufficient for\;} g\}$ is convex. 
\end{proposition}

\begin{proof}If the set is either empty or a singleton, there is nothing to prove. So, assume that the set contains at least two elements, say  $\delta$ and $\widehat{\delta}$. Since $f$ is $\delta$-sufficient for $g$, there exist garblings $(\gamma_t)$ such that  \begin{align*}
\delta_t g^t(y^t|\theta) = \sum_{t'} \delta_{t'} f^{t'}(x^{t'}|\theta) \gamma_{t'}(y^t|x^{t'}), \text{for all\;} (t,y^t,\theta).
\end{align*}
Similarly, since $f$ is $\widehat{\delta}$-sufficient for $g$, there exist garblings $(\widehat{\gamma}_t)$ such that 
\begin{align*}
\widehat{\delta}_t g^t(y^t|\theta) = \sum_{t'} \widehat{\delta}_{t'} f^{t'}(x^{t'}|\theta) \widehat{\gamma}_{t'}(y^t|x^{t'}), \text{for all\;} (t,y^t,\theta).
\end{align*}
Let $\alpha \in [0,1]$ and consider the discount factor $\alpha \delta + (1-\alpha) \widehat{\delta}$. It is routine to verify that the garblings 
\begin{align*}
\left(\frac{\alpha \delta_t}{\alpha \delta_t + (1-\alpha) \widehat{\delta}_t}\gamma_t + \frac{(1-\alpha) \widehat{\delta}_t}{\alpha \delta_t + (1-\alpha) \widehat{\delta}_t} \widehat{\gamma_t}\right)_t
\end{align*}
satisfy the required equalities for $\alpha \delta + (1-\alpha) \widehat{\delta}$-sufficiency. 
\end{proof}

We now present two applications of Proposition \ref{prop:convexity}. As a first application, consider two discount factors  $\underline{\delta}$ and $\overline{\delta}$ such that  $\overline{\delta}$ first-order stochastically dominates $\underline{\delta}$. We interpret first-order stochastic shifts in discount factors as the decision-maker becoming more patient. For instance, if we restrict attention to geometric discounting, first-order stochastic shifts correspond to higher discount rates. If $f$ is $\underline{\delta}$ and $\overline{\delta}$-sufficient for $g$, it is natural to query whether $f$ is $\delta$-sufficient for $g$ for all $\delta$, which first-order stochastically dominate $\underline{\delta}$ and are dominated by $\overline{\delta}$. An immediate consequence of Proposition \ref{prop:convexity} is that the statement is true when $T=2$ since $\underline{\delta}$ and $\overline{\delta}$ are the extreme points of the first-order stochastic dominance intervals.\footnote{The interval is $[\underline{\delta}_1,\overline{\delta}_1]$.} With more periods, it is no longer true that $\underline{\delta}$ and $\overline{\delta}$ are the extreme points. For instance, if $\underline{\delta}= (1,0,0)$ and $\overline{\delta}=(0,0,1)$, then the discount factor $\delta=(0,1,0)$ first-order stochastically dominates $\underline{\delta}$, is dominated by $\overline{\delta}$, and cannot be written as  convex combination of the two of them.  (See Theorem 1 in \cite{Yang:2023aa} for a characterization of the extreme points of first-order stochastic dominance intervals.) Yet, this does not immediately imply that $f$ is not $\delta$-sufficient for $g$. The following example, however, shows this is the case.
\medskip

\textbf{Example 2.}  There are two states $\theta_0$ and $\theta_1$, three periods, and two experiments $f$ and $g$.  Let $p \in (0,1)$  be the prior belief of $[\boldsymbol{\theta}=\theta_0]$.  We describe the evolution of the beliefs induced by $f$ and $g$ in Figure \ref{fig:counter-ex}. (It is straightforward to construct the associated signals $X_t$ (resp., $Y_t$) and kernels $f_t$ (resp., $g_t$).) 

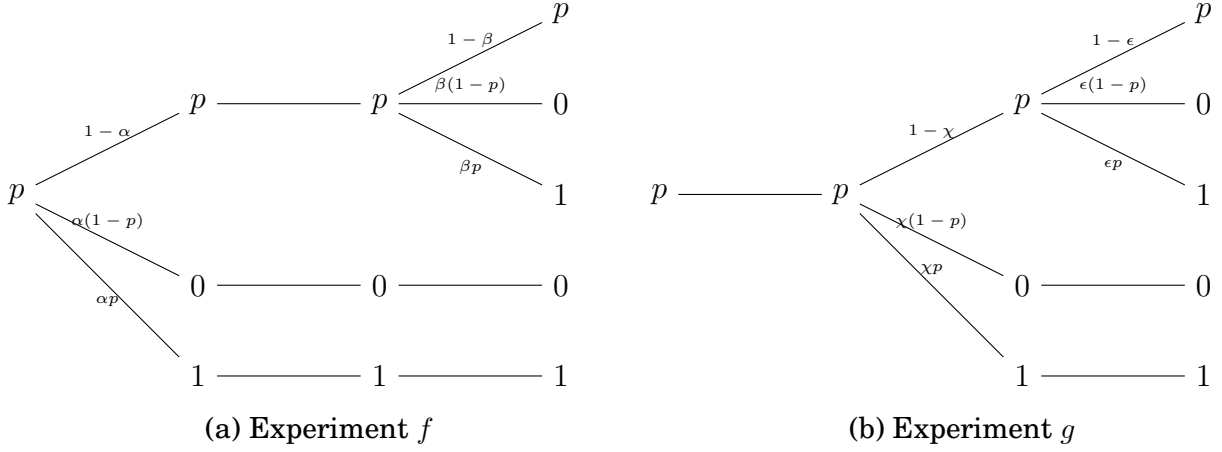
\begin{figure}[h]
\centering
\begin{subfigure}{.5\textwidth}
\begin{tikzpicture}[scale=1.2]
  \node (1) at (0,3) {$p$};
  \node(2) at (2,4) {$p$};
  \node(3) at (2,2) {$0$};
  \node(4) at (2,1) {$1$};
  \node(5) at (4,4) {$p$};
  \node(6) at (4,2) {$0$};
  \node(7) at (4,1) {$1$};
  \node(8) at (6,5) {$p$};
  \node(9) at (6,4) {$0$};
  \node(10) at (6,3) {$1$};
  \node(11) at (6,2) {$0$};
  \node(12) at (6,1) {$1$};
  \draw (1) -- (2) node[midway, above] { \tiny{$1-\alpha$}}; 
  \draw (2)  -- (5) ;
  \draw (5) -- (8) node[midway, above] { \tiny{$1-\beta$}};
  \draw (1) -- (3) node[midway, above ] { \tiny{$\alpha (1-p)$}};
  \draw (3) -- (6) ;
  \draw (6) -- (11);
  \draw (1) -- (4) node[midway, below ] { \tiny{$\alpha p$}}; 
  \draw (4) -- (7) ; 
  \draw (7) -- (12);
  \draw (5) -- (9) node[midway, above ] { \tiny{$\beta (1-p)$}};
  \draw (5) -- (10) node[midway, below ] { \tiny{$\beta p$}} ;
\end{tikzpicture}
\caption{Experiment $f$}
\end{subfigure}%
\begin{subfigure}{.5\textwidth}
\begin{tikzpicture}[scale=1.2]
  \node (1) at (0,3) {$p$};
  \node(2) at (2,3) {$p$};
  \node(3) at (4,4) {$p$};
  \node(4) at (4,2) {$0$};
  \node(5) at (4,1) {$1$};
  \node(6) at (6,5) {$p$};
  \node(7) at (6,4) {$0$};
  \node(8) at (6,3) {$1$};
  \node(9) at (6,2) {$0$};
  \node(10) at (6,1) {$1$};
  \draw (1) -- (2) ; 
  \draw (2)  -- (3) node[midway, above] { \tiny{$1-\chi$}} ;
  \draw (2) -- (4) node[midway, above] { \tiny{$\chi (1-p)$}};
  \draw (2) -- (5) node[midway, above] { \tiny{$\chi p $}} ;
  \draw (4) -- (9) ;
  \draw (5) -- (10) ;
  \draw (3) -- (6) node[midway, above] { \tiny{$1-\epsilon$}};
  \draw (3) -- (7)  node[midway, above ] { \tiny{$\epsilon (1-p)$}} ;
  \draw (3) -- (8)  node[midway, below ] { \tiny{$\epsilon p$}};
  \end{tikzpicture}
\caption{Experiment $g$}
\end{subfigure}
\caption{Discounting, first-order stochastic dominance, and sufficiency}\label{fig:counter-ex}
\end{figure}

The experiment $f$ is informative at the first and third periods, while $g$ is informative at the second and third periods. Let $v: [0,1] \rightarrow \mathbb{R}$ be a convex function and denote $V(p):=(1-p)v(0) + pv(1) \geq v(p)$. As already explained, $\delta$-sufficiency is equivalent to checking whether the discounted expected value $V_{f}(\delta,v)$ is greater than $V_{g}(\delta,v)$ for all convex functions $v$, where   
\begin{align*}
V_{f}(\delta,v)& :=v(p) \left[ (\delta_1+\delta_2)(1-\alpha) + \delta_3 (1-\alpha)(1-\beta) \right] + V(p)\left[(\delta_1 +\delta_2)\alpha+ \delta_3 (\alpha + (1-\alpha)\beta) \right],\\
V_{g}(\delta,v) & :=v(p) \left[ \delta_1+\delta_2(1-\chi) + \delta_3 (1-\chi)(1-\varepsilon) \right] + V(p)\left[(\delta_2 +\delta_3)\chi + \delta_3 (1-\chi)\varepsilon \right].
\end{align*}
Clearly, if $v(p)=V(p)$, $V_{f}(\delta,v) =V_{g}(\delta,v)$ for all $\delta$. Assume that $V(p) >v(p)$. Comparing $V_{f}(\delta,v)$ and $V_{g}(\delta,v)$
is thus equivalent to comparing the coefficients in front of $V(p)$ in the above expression. Consider the following (geometric) discount factors:  $\underline{\delta} = (1,0,0)$, $\overline{\delta} = (1/3,1/3,1/3)$ and $\delta^*= (4/7,2/7,1/7)$ and observe that $V_{f}(\delta,v) > V_{g} (\delta,v)$ is equivalent to 
\begin{align*}
(\delta_1 +\delta_2)\alpha+ \delta_3 (\alpha + (1-\alpha)\beta) > (\delta_2 +\delta_3)\chi + \delta_3 (1-\chi)\varepsilon. 
\end{align*}
We can verify that with the parametrization $\varepsilon =0$, $\beta=1$, $\alpha = 0 $, $\chi \in (1/3,1/2)$, $V_{f}(\underline{\delta},v)= V_{g}(\underline{\delta},v)$, $V_{f}(\delta^*,v) <  V_{g}(\delta^*,v)$, and  $V_{f}(\overline{\delta},v) > V_{g}(\overline{\delta},v)$ for all convex functions $v$ such that $V(p) >v(p)$. With an arbitrary small, but positive, $\alpha$, we have that $V_{f}(\underline{\delta},v) > V_{g}(\underline{\delta},v)$, without affecting the other comparisons. So, the experiment $f$ is more informative than $g$ for a very patient and impatient decision-maker, and the converse holds for a moderately patient decision-maker. In other words, we have ambiguous comparative statics with respect to impatience.

\medskip

As a second application, we consider sets of discount factors. In repeated games and, more generally, in economics, it is frequent to state results for sufficiently patient decision-makers, without fixing the discount factor. We would like the experiment $f$ to be more valuable than the experiment $g$ in all discounted problems with discount factors in $\Delta$, a non-empty set of discount factors. For instance, $\Delta$ may be the set of geometric discount factors. Theorem \ref{th:fixed-discounting-uncontrolled} generalizes immediately.   All we need is to add ``for all $\delta \in \Delta$'' in the second and third statements.  In particular, if we consider the set of all discount factors, we obtain a  simple characterization, that is, we need $f^t$ to be sufficient for $g^t$ for all $t$. The intuition is straightforward. If $\delta$  is degenerated in $t$, then $f$  being  $\delta$-sufficient for $g$ is clearly equivalent to $f^t$ being sufficient for $g^t$. Since any discount factor is a convex combination of the degenerated probabilities on integers $t$,  the argument follows immediately from Proposition \ref{prop:convexity}.

\subsection{Time-varying States} This section relaxes the assumption of a fixed state. We show that our analysis generalizes almost immediately to the state evolving over time. We again limit the discussion to the setting of subsection \ref{sec:basics}. Let  $p \in \Delta(\Theta_1 \times \dots \times \Theta_T)$ be the distribution over sequences of states. Without loss of generality, we assume that $\Theta_t=\Theta$ for all $t$. We assume furthermore that the decision-maker's payoff at period $t$ only depends on the state at period $t$. (If the payoff was a function of the entire profile of states, then the previous analysis applies verbatim.) We write $\mathbb{P}_{p,f,\sigma}$ for the distribution over states, signals and actions, induced by the distribution of states $p$, the experiment $f$ and the strategy $\sigma$. Similarly, for   $\mathbb{P}_{p,g,\tau}$.

As in Section \ref{sec:main}, we can rewrite the expected payoff as: 
\begin{align*}
\mathbb{E}_{p, g, \tau}\left[\sum_{t}\delta_t u(\boldsymbol{a}_t,\boldsymbol{\theta}_t)\right]  & = \sum_{a,\theta} u(a,\theta)\underbrace{\left[\sum_t \delta_t \marg_{A_t \times \Theta_t}\proba_{p,g,\tau}(a,\theta)\right]}_{\text{discounted probability of $(a,\theta)$}}.
\end{align*}
It once again follows immediately that an equivalent statement to $f$ being more valuable than $g$ is that for all decision problems $(A,u)$, for all strategy $\tau$, there exists strategy $\sigma$ such that
\begin{align*}
\sum_{t}\delta_t \mathbb{P}_{p,g,\tau}(\boldsymbol{a}_t=a,\boldsymbol{\theta}_t=\theta) = \sum_{t}\delta_t \mathbb{P}_{p,f,\sigma}(\boldsymbol{a}_t=a,\boldsymbol{\theta}_t=\theta),
\end{align*}
for all $(a,\theta)$. In words, all we need to replicate are the discounted probabilities of states and actions. \medskip 

Since neither the realizations of the states nor the realizations of the signals depend on the decisions made, the distribution over states and signals induced by $p$, $g$ and $\tau$ process is independent of $\tau$. We denote it $\mathbb{P}_{p,g}$. In particular, if the state is fully persistent, i.e., $p(\theta_1,\dots,\theta_t,\dots,\theta_T)>0$ if, and only if, $\theta_t=\theta_{t'}$ for all $(t,t')$, then 
\begin{align*}
\mathbb{P}_{p,g}(\boldsymbol{y}^t=y^t,\boldsymbol{\theta}_t=\theta) = p(\theta,\dots,\theta,\dots,\theta)g^t(y^t|\theta).
\end{align*}
We define  $\mathbb{P}_{p,f}$ similarly. The following theorem is the generalization of Theorem \ref{th:fixed-discounting-uncontrolled} to an environment, with evolving states.

\begin{theorem}\label{th:evolving-states} Let $p$ be the distribution over the sequences of states. The following statements are equivalent. 
\begin{enumerate}
\item The experiment $f$ is more valuable than the experiment $g$ in discounted problems with discount factor $\delta$.
\item For all decision problem $(A,u)$, for all  
strategy $\tau$, there exists a strategy $\sigma$ such that 
\[\sum_{t}\delta_t\marg_{A_t \times \Theta_t}\mathbb{P}_ {p, g, \tau} = \sum_{t} \delta_t\marg_{A_t \times \Theta_t} \mathbb{P}_{p, f,\sigma}.\] 
\item $\delta$-sufficiency: There exist garblings $\gamma_{t'}: X^{t'} \rightarrow \Delta(\cup_{t=1}^{T}Y^t)$ such that 
\begin{align*}
\delta_t \mathbb{P}_{p,g}(\boldsymbol{y}^t=y^t,\boldsymbol{\theta}_t=\theta)= 
\sum_{t',x^{t'}} \delta_{t'}   \mathbb{P}_{p,f}(\boldsymbol{x}^{t'}=x^{t'},\boldsymbol{\theta}_{t'}=\theta)\gamma_{t'}(y^t|x^{t'}), \text{for all\;} t,y^t,\theta.
\end{align*}
\end{enumerate}
\end{theorem}

\subsection{Sequential Comparison}\label{subsec:seq-compa} The comparisons we have made so far are from the ex-ante viewpoint, that is, the decision-maker chooses 
whether to observe the  signals from either $f$ or $g$ at the first stage and cannot revise his choice at a later date. We want to strengthen our comparison to give the decision-maker the opportunity  to \emph{permanently} switch from observing the $X$-signals to observing the $Y$-signals after any history $x^t$ of $X$-signals. In other words, we want the decision-maker to prefer observing signals from $f$ not only at the ex-ante stage, but also at all (on-path) histories. The following example illustrates the issues we have to deal with. \medskip

\textbf{Example.} There are two periods, two states $L$(ow) and $H$(high), and two signals $\ell$(ow) and $h$(igh). The experiments $f$ and $g$ are given by:
\begin{align*}
f_1(\ell|L)=f_1(h|H) =3/4, f_2(h|h,\theta)=f_2(\ell|\ell,\theta)=1, \\
g_1(\ell|L)=g_1(h|H) =3/4, g_2(h|h,\theta)=g_2(\ell|\ell,\theta)=1,
\end{align*}
with $\theta=L,H$. The two experiments have the same probability law and simply repeat the first-period signal at the second period.  As random processes, they may differ, however. If they are independent, after observing the first-period signal from either of them, the decision-maker always has an incentive to switch. By doing so, the decision-maker obtains two independent signals drawn from the same distribution instead of a single one. If the two processes are perfectly correlated, the decision-maker never has an incentive to switch. The example illustrates the need to model the interdependence between the experiments if we want to formalize the incentives of the decision-maker to switch from one experiment to another. \medskip

The approach we follow is to consider a joint process (a coupling) over the $X$- and $Y$-signals such that marginal over the $X$- (resp., $Y$-) signals coincides with the experiment $f$ (resp., $g$). The joint process allows us to define the ``continuation experiments,'' i.e., the experiments the decision-maker faces when he considers switching from one experiment to another after having observed the history of signals $x^t$. More precisely, the joint process is given by the collection of kernels $(h_t)_t$, where
\begin{align*}
h_t: \Theta \times X^{t-1} \times Y^{t-1} \rightarrow \Delta(X_t \times Y_t).
\end{align*}
For all $x^t \in X^t$, we then define the experiments 
\begin{align*}
f[x^t] &: \Theta \rightarrow \Delta(X_{t+1} \times \dots \times X_T), \\
g[x^t] & : \Theta \rightarrow \Delta(Y_{t+1} \times \dots \times Y_T),
\end{align*} 
where 
\begin{align*}
f[x^{t}](x_{t+1},\dots,x_T|\theta)  :=
\sum_{y^{t}}h(x_{t+1}, \dots, x_T|x^{t},y^{t},\theta)h^{t}(y^{t}|x^{t},\theta),
\end{align*}
whenever $h^t(x^t|\theta)>0$ (and arbitrary, otherwise). Similarly, for $g[x^t]$. Thus, after observing the sequence of signals $x^t$, the decision-maker can either continue observing the signals from $f[x^t]$ or switch to observing the signals  from $g[x^t]$. Note that $f$ and $g$ correspond to  $f[\emptyset]$ and $g[\emptyset]$, respectively. We want the decision-maker to find $f[x^t]$ more valuable than $g[x^t]$ at all $x^t$. In what follows, we denote $\tau[x^t]$ and $\sigma[x^t]$ the strategies of the decision problem starting at period $t+1$, adapted to $g[x^t]$ and $f[x^t]$, respectively.

\begin{definition}\label{def:seq-comp}
Let $\delta$ be the discount factor. The  experiment $f$ is \emph{sequentially most valuable} for $h$ in discounted problems, if for all decision problems $(A, u)$, for all history of signals $x^t$ (including the empty history), for all strategy $\tau[x^t]$, there exists strategy $\sigma[x^t]$ such that
\begin{align*}
\mathbb{E}_{\theta,g[x^t],\tau[x^t]}\Big[\sum_{t'\geq t+1} \delta_{t'} u(\boldsymbol{a}_{t'},\theta)\Big] \leq 
\mathbb{E}_{\theta,f[x^t],\sigma[x^t]}\Big[\sum_{t' \geq t+1} \delta_{t'} u(\boldsymbol{a}_{t'},\theta)\Big]
\end{align*}
for all $\theta$ such that $h(x^t|\theta) >0$.
\end{definition}
The definition states that the experiment $f$ is sequentially most valuable if the decision-maker not only prefers observing the $X$-signals at the empty history, but  also at any (on-path) history $x^t$. Several  modifications of this definition are possible. For instance, we may want $f$ to be sequentially most valuable not only for a fixed $h$, but also for all joint processes $h'$ such that the marginal over the $X$-signals (resp., $Y$-signals) is $f$ (resp., $g$).  As another instance, we may want to compare the experiments not only at on-path histories, but also at off-path histories. As yet another, we may want to give the decision-maker the opportunity to freely choose the signal to observe at each history. We view our definition as a natural starting point for the sequential comparison of experiments.\medskip

Given the similarity between the ex-ante comparison of experiments (Definition \ref{def:valuable-discounted}) and their sequential comparison (Definition \ref{def:seq-comp}), it is natural to expect their characterizations to be similar.  As Theorem \ref{th:seq-com} demonstrates, it is indeed the case. To state the theorem, we need to introduce  some notation. For all $x^t$, we let $\Theta[x^t] :=\{\theta \in \Theta: h(x^t|\theta)>0\}$, $\delta[t]: = (\delta_{t+1}/(\sum_{t'\geq t+1}\delta_t),\dots,)$, and say that $f[x^t]$ is $\delta[t]$-sufficient for $g[x^t]$  if $g[x^t]: \Theta[x^t] \rightarrow \Delta(Y_{t+1} \times \dots \times Y_T)$ is a garbled version of $f[x^t] : \Theta[x^t] \rightarrow \Delta(X_{t+1} \times \dots \times X_T)$. In words, $\Theta[x^t]$ is the set of states, which have not been excluded yet, and we compare the experiments $f[x^t]$ and $g[x^t]$ only for these states. The discount factor $\delta[t]$ is the normalization of $\delta$ to the periods $t+1$ to $T$.

\begin{theorem}\label{th:seq-com}
The experiment $f$ is sequentially most valuable for $h$ if, and only if, $f[x^t]$ is $\delta[t]$-sufficient for $g[x^t]$ (with respect to $\Theta[x^t]$) for all $x^t$.
\end{theorem}

We conclude with a discussion of the recent work of \cite{BFK-2023} on the comparison of experiments. These authors study when the decision-maker prefers observing the $X$-signals over the $Y$-signals when he has already observed some $Z$-signals. They adopt a different notion of comparison than ours in that they require the decision-maker to prefer observing the $X$-signals, regardless of the $Z$-signals (and, therefore, of their correlations with the $X$- and $Y$-signals).\footnote{\cite{BFK-2023} model experiments differently, but this is not essential for our discussion.} We can view their model as a two-period problem, where   $h_1: \Theta \rightarrow \Delta(X_1 \times Y_1)$, $X_1=Y_1=Z$, $h_1(x_1,y_1|\theta)>0$ only if $x_1=y_1$, $h_2: \Theta \times X_1 \times Y_1 \rightarrow \Delta(X_2 \times Y_2)$, and $\delta=(0,1)$. In this two-period problem, the experiment $f$ is sequentially most valuable for $h$ if the decision-maker prefers observing the signal $x_2$ to $y_2$ in the second period, regardless of the realized signal $z=x_1=y$ in the first period. This is weaker than the version of \cite{BFK-2023} since we fix $h$.

\section{Related Literature and Discussion}\label{sec:rel-litt}

This paper is closely related to the extensive literature on the comparison of statistical experiments in static problems. We refer the reader to  \citet[Chapter 12]{blackwell-girshick}, \cite{torgersen91}, \cite{lecam96} and    \citet[Chapter 9]{strasser2011}, among others, for extensive reviews. As we argue in Section \ref{sec:main}, there is a tight relationship between the simplest discounted problems and the static problems. Indeed, the condition of $\delta$-sufficiency coincides with the classical sufficiency condition \citep{blackwell51, blackwell53} for the comparison of the mixture experiments $\sum_t \delta_t f^t$ and $\sum_t \delta_t g^t$.  Few results are available for the comparison of mixture experiments. \cite{torgersen1970} proves that if $f^t$ is sufficient for $g^t$ for all $t$, then $\sum_t \delta_t f^t$ is sufficient for $\sum_t \delta_t g^t$.  He also proves a partial converse, that is, if the experiment $f$ is sufficient for $\sum_t \delta_t g^t$, then there exist experiments $f^t$ such that (i) $f = \sum_t \delta_t f^t$, and (ii) $f^t$ is sufficient for $g^t$ for all $t$. A related result is that all binary experiments (a.k.a dichotomies) are the mixtures of fully informative experiments -- see \cite{birnbaum1961}. The relationship with static problems becomes less tight as we move from the simplest discounted problems to more general ones. For instance, static problems do not allow discussing about the sequential comparisons of experiments, evolving states, or past decisions influencing the evolution of the signals.\medskip 

By comparison, the literature on the comparison of statistical experiments in dynamic problems is much more limited. The seminal paper is \cite{greenshtein96}. Greenshtein considers the larger class of sequential problems, where the decision-maker is constrained to choose sequences of decisions in $\mathcal{A} \subseteq A_1 \times \dots \times A_t \times \dots \times A_{T}$, and his payoff is  an arbitrary function of the sequences of decisions made, that is, $u: \bigtimes_t A_t \times \Theta \rightarrow \mathbb{R}$. \citet[Theorem 1.1a]{greenshtein96} proves that the experiment $f$ is more valuable than the experiment $g$ for all sequential problems if, and only if, $g$ is a garbling of $f$, where the garbling is \emph{adapted}.   Informally, a garbling is adapted if, to simulate $g^t$, the decision-maker first draws $y_1$ with probability $\gamma_1(y_1|x_1)$, next draws $y_2$ with probability $\gamma_2(y_2|x^2,y^1)$, and then iteratively draws $y_{t'}$ with probability $\gamma_{t'}(y_{t'}|y^{t'-1},x^{t'})$ for all $t' \leq t$ when his history of simulated and realized signals is $(y^{t-1},x^t)$. In other words, unlike what we did, the decision-maker does not simulate new fictitious histories of signals at each period. Instead, at each period, he completes the history of already simulated signals. This is a measurability condition. With the help of an example, \cite{greenshtein96} shows that the condition is stronger than requiring $f^t$ to be sufficient for $g^t$, for all $t$.\footnote{We have generalized the work of \cite{greenshtein96} to sequential problems, where the decision-maker controls the flow of information. The statement and proof are available upon request.}$^{,}\,$\footnote{\cite{greenshtein96} also provides conditions linking the two notions of sufficiency.}
For completeness, we provide a characterization of the latter condition. A sequential problem is \emph{separable}  if the decision-maker values the sequences of actions $(a_1,\dots,a_T)$ as $\sum_t u_t(a_t,\theta)$. We note that discounted problems are special cases of separable problems.

\begin{theorem}\label{th:separable-uncontrolled} The following statements are equivalent: 
\begin{enumerate}[label=(\arabic*)]
\item The experiment $f$ is more valuable than the experiment $g$ in all separable problems.
\item For all separable decision problems $(A_t,u_t)_t$, for all  
strategy $\tau$, there exists strategy $\sigma$ such that 
\begin{align*}
\marg_{A_t}\mathbb{P}_ {\theta, g, \tau} = \marg_{A_t} \mathbb{P}_{\theta, f,\sigma}, \text{\;for all\;} t, \theta.
\end{align*}
\item  $\Delta$-sufficiency: There exist garblings $\gamma_t: X^t \rightarrow \Delta(Y^t)$ such that 
\begin{align*}
g^t(y^t|\theta)= \sum_{x^t} \gamma_t(y^t|x^t) f^t(x^t|\theta) , \text{\;for all\;} y^t,t,\theta.
\end{align*}
\end{enumerate}
\end{theorem} 
\medskip 

As the introductory example demonstrates, the condition of $\Delta$-sufficient is stronger than the condition of $\delta$-sufficiency. However, it follows from the discussion in Section \ref{sec:discount factors} that requiring $\delta$-sufficiency for all discount factors is equivalent to $\Delta$-sufficiency. We record this result in the following corollary. 
\begin{corollary}
The following statements are equivalent
\begin{enumerate}[label=(\arabic*)]
\item The experiment $f$ is $\Delta$-sufficient for the experiment $g$.
\item The experiment $f$ is $\delta$-sufficient for the experiment $g$ for all $\delta$. 
\end{enumerate}
\end{corollary}

Also, unlike \cite{greenshtein96}, we have not studied constrained discounted problems. (A problem is constrained if the set of feasible decisions at a period depends on the history of decisions made in the past.) We do so in the appendix and prove an equivalence with the result of  \cite{greenshtein96}, that is, the experiment $f$ is more valuable than the experiment $g$ in all constrained \emph{discounted} problems if, and only if, it is more valuable in all constrained \emph{sequential} problems.  Intuitively, we have so much flexibility in the choice of feasible sequences of decisions that we can force the decision-maker to complete histories of simulated signals, as in \cite{greenshtein96}.\medskip 

As already discussed, in a concurrent paper, \cite{whitmeyer-williams-2024} prove an equivalent version of Theorems \ref{th:fixed-discounting-uncontrolled} and  \ref{th:separable-uncontrolled} in terms of distributions over posterior beliefs. Neither \cite{greenshtein96} nor \cite{whitmeyer-williams-2024} extend their analysis to, among others, the sequential comparison of experiments, time-varying states or controlled flow of information, as we do.

We conclude with a discussion of few other related works. \cite{liang2022} study the optimal allocation of resources across multiple sources of information in a continuous-time stopping problem. (See also \cite{liang2018}.)  More precisely, there is fixed payoff-relevant state given by a linear combination of unknown parameters $(\theta_1,\dots,\theta_n)$, with a known Gaussian distribution. At each instant, the decision-maker either allocates his (unit) budget to the $n$ sources of information or takes a decision -- taking a decision stops the game. The budget allocation influences the diffusion process the decision-maker observes.\footnote{In discrete time, if the decision-maker allocates $q_i$ to source $i$, he observes the noisy signal $\theta_i + \varepsilon_i$, where $\varepsilon_i$ is Gaussian with zero mean and variance $1/q_i$. Their modeling is the continuous-time analogue.} Under some assumptions, the authors show that the optimal strategy is history-independent and consists in first allocating the entire budget to the most informative source and, subsequently, to invest into additional sources until all sources are funded  -- investments are made at fixed times. Moreover, their analysis shows that the strategy is optimal regardless of the decision problem the decision-maker faces when he stops learning.  In other words, their strategy induces an experiment, which is sufficient for all other inducible experiments.  In that sense, their work generalizes Theorem 3.1 of \cite{greenshtein96}.\footnote{It is not a generalization of Theorem 1.1a, which is the one we discussed earlier.} Our work differs from theirs in two important aspects. First, we consider different classes of games and do not restrict attention to Gaussian distributions. Second, there is a natural decoupling between the decision problem the decision-maker faces when he stops learning and the choice of experiments (induced by the budget allocation). Such a decoupling is not always possible in our analysis -- see Section \ref{sec:control-info}. It is also worth noting that since the decision-maker can allocates his budget at each period, their comparison is sequential.\medskip 

 \cite{lehrer2023} also study a stopping problem, where the decision-maker either acquires some information or stops and takes an action. At period $t$, if the decision-maker decides to acquire some information, he pays a fee and obtains a signal from a static experiment $f$. (Signals are drawn independently and identically across time.)  Their main focus is on how an information provider should design the fee structure. They show that upfront fees are optimal. Intuitively, upfront fees do not distort the incentives to acquire information and insure that the decision-maker purchases as often as possible the information for sale. They also show that if $f$ is sufficient for $g$, acquiring information is optimal for the decision-maker under a larger set of beliefs when the information is from $f$.   This does not mean, however, that the decision-maker acquires information for longer as beliefs also get updated faster when the information is from the most informative experiment $f$. This latter observation is reminiscent of Example 2.1 in \cite{greenshtein96}. \medskip
 
 Finally,  \cite{mu2021} study when $T$ independent and identical copies of one static experiment  is (eventually) sufficient for $T$ independent and identical copies of another.\footnote{More precisely, they require the existence of $\overline{T}$ such that for all $T \geq \overline{T}$, $T$ independent and identical copies of one static experiment is sufficient for another.} They provide a complete characterization for dichotomies in terms of the Renyi divergence of the static experiments. For earlier related results, see \cite{torgersen1970},  \cite{torgersen91}, and      \cite{moscarini2002}. \medskip

\newpage
	\appendix
	\section{Proofs}
	
\subsection{Proof of Theorem \ref{th:fixed-discounting-uncontrolled} and \ref{th:evolving-states}} We first prove Theorem \ref{th:evolving-states} and then derive Theorem \ref{th:fixed-discounting-uncontrolled} as a corollary. \medskip 

$(2) \implies (1)$. Immediate. \medskip 

$(1) \implies (2)$.  Consider two behavioral strategies $\sigma$ and $\sigma'$, with induced distributions $\mathbb{P}_{p,f,\sigma}$ and $\mathbb{P}_{p,f,\sigma'}$, respectively. From Kuhn's theorem, there exist two mixed strategies $\widehat{\sigma}$ and $\widehat{\sigma}'$, which induce the same distributions over states, signals and actions.  From the convexity of the set of mixed strategies and the linearity of the measures, 
\[ \mathbb{P}_{p,f,\alpha\widehat{\sigma}+ (1-\alpha)\widehat{\sigma}'} = \alpha \mathbb{P}_{p,f,\widehat{\sigma}} + (1-\alpha) \mathbb{P}_{p,f,\widehat{\sigma}'}\]
for all $\alpha \in [0,1]$. Therefore, the set 
\[\left\{\sum_{t} \delta_t \marg_{A_t \times \Theta_t}\mathbb{P}_{p,f,\widehat{\sigma}}: \widehat{\sigma} \text{\;a mixed strategy}\right\} \subseteq \Delta(A \times \Theta)\]
is convex. (From Kuhn's theorem, we obtain the same set if we consider all behavioral strategies.)

If (2) does not hold, there exists a behavioral strategy $\tau$ and, therefore, a mixed strategy $\widehat{\tau}$ such that 
\[\sum_{t} \delta_t \marg_{A_t \times \Theta_t}\mathbb{P}_{p,g,\widehat{\tau}} \notin  \left\{\sum_{t} \delta_t \marg_{A_t \times \Theta_t}\mathbb{P}_{p,f,\widehat{\sigma}}: \widehat{\sigma} \text{\;a mixed strategy}\right\}.\]
As in \cite{blackwell51}, we can then apply a separation argument to prove the existence of a payoff function $u$ such that $(1)$ is violated, the required contradiction.

\medskip 

$(3) \implies (2)$. Fix a decision problem $(A,u)$ and a strategy $\tau$. Let $(\gamma_{t'})_{t'}$ be the garblings, with $\gamma_{t'}: X^{t'} \rightarrow \Delta(\bigcup_{t}Y^{t})$. Note that summing over all elements in $\bigcup_{t}Y^{t}$ is the same as summing over all pairs $(t, y^{t})$, where the second coordinate is  in $Y^{t}$. Throughout, we therefore indicate summations over $\bigcup_{t}Y^{t}$ by summing over $(t,y^t)$. For all $t'$, for all $x^{t'}$, for all $a$, define $\sigma_{t'}$ as:
\begin{align*}
\sigma_{t'}(a|x^{t'}) := \sum_{(t,y^{t})}\mathbb{P}_{p,g,\tau}(\boldsymbol{a}_{t}=a|y^{t}) \gamma_{t'}(y^{t}|x^{t'}).
\end{align*}
From the above observation, this is well-defined as $\sum_{(t,y^{t})}\gamma_{t'}(y^{t}|x^{t'}) =1$.

Under the strategy $\sigma = (\sigma_{t'})_{t'}$, the discounted probability of choice $a$ is 
\begin{align*}
\sum_{t'}\delta_{t'} \mathbb{P}_{p,f,\sigma}(\boldsymbol{a}_{t'}=a,\boldsymbol{\theta}_{t'}=\theta) = 
\sum_{t',x^{t'}}\delta_{t'} \mathbb{P}_{p,f,\sigma}(\boldsymbol{a}_{t'}=a|x^{t'},\boldsymbol{\theta}_{t'}=\theta)\mathbb{P}_{p,f,\sigma}(x^{t'},\boldsymbol{\theta}_{t'}=\theta) = \\
\sum_{t',x^{t'}}\delta_{t'} \mathbb{P}_{p,f,\sigma}(x^{t'},\boldsymbol{\theta}_{t'}=\theta)\sigma_{t'}(a|x^{t'})= \\
\sum_{t',x^{t'}}\delta_{t'} \mathbb{P}_{p,f,\sigma}(x^{t'},\boldsymbol{\theta}_{t'}=\theta)\Bigl(\sum_{t,y^{t}}\mathbb{P}_{p,g,\tau}(\boldsymbol{a}_{t}=a|y^{t}) \gamma_{t'}(y^{t}|x^{t'})\Bigr) = \\
\sum_{t,y^{t}}\mathbb{P}_{p,g,\tau}(\boldsymbol{a}_{t}=a|y^{t})\Bigl(\sum_{t',x^{t'}}\delta_{t'}\mathbb{P}_{p,f,\sigma}(x^{t'},\boldsymbol{\theta}_{t'}=\theta) \gamma_{t'}(y^{t}|x^{t'})\Bigr) = \\
\sum_{t,y^{t}}\mathbb{P}_{p,g,\tau}(\boldsymbol{a}_{t}=a|y^{t}) \delta_{t}\mathbb{P}_{p,g,\tau}(y^{t},\boldsymbol{\theta}_{t}=\theta) = \\
\sum_{t} \delta_{t}\mathbb{P}_{p,g,\tau}(\boldsymbol{a}_{t}=a, \boldsymbol{\theta}_{t}=\theta),
\end{align*}
where
\begin{align*}
\delta_{t} \mathbb{P}_{p,g,\tau}(y^{t},\boldsymbol{\theta}_{t}=\theta) = \sum_{t',x^{t'}} \delta_{t'} \mathbb{P}_{p,f,\sigma}(x^{t'},\boldsymbol{\theta}_{t'}=\theta)\gamma_{t'}(y^{t}|x^{t'}),
\end{align*}
from $\delta$-sufficiency. This follows from the observation that 
\begin{align*}
\mathbb{P}_{p,f,\sigma}(x^{t'},\boldsymbol{\theta}_{t'}=\theta) & =  \sum_{\theta_1,\dots,\theta_{t'},\dots,\theta_T : \theta_{t'}=\theta}p(\theta_1,\dots,\theta_{t'},\dots,\theta_T)f^{t'}(x^{t'}|\theta_1,\dots,\theta_{t'},\dots,\theta_T),\\
\mathbb{P}_{p,g,\tau}(y^t,\boldsymbol{\theta}_t=\theta) & =  \sum_{\theta_1,\dots,\theta_t,\dots,\theta_T : \theta_t=\theta}p(\theta_1,\dots,\theta_t,\dots,\theta_T)g^t(y^t|\theta_1,\dots,\theta_t,\dots,\theta_T).
\end{align*}
\medskip 
	
$(2) \implies (3)$. Let $A = \bigcup_{t=1}^{T}Y^{t}$ and $\tau_{t}(y^{t} |y^{t}, a^{t-1})=1$ for all $(t,y^{t},a^{t-1})$. By construction, 	
\begin{align*}	
\sum_{t'}\delta_{t'}\mathbb{P}_{p,g,\tau}(\boldsymbol{a}_{t'}=y^{t},\boldsymbol{\theta}_{t'}=\theta) =  \delta_{t}	\sum_{\theta_1,\dots,\theta_{t},\dots,\theta_T : \theta_{t}=\theta}p(\theta_1,\dots,\theta_{t},\dots,\theta_T)g^{t}(y^{t}|\theta_1,\dots,\theta_{t},\dots,\theta_T).
\end{align*}
From (2), there exists $\sigma=(\sigma_{t'})_{t'}$, where $\sigma_{t'}: X^{t'} \rightarrow \Delta( \bigcup_{t}^{T}Y^{t})$. such that 
\[\sum_{t'}\delta_{t'}\mathbb{P}_{p,g,\tau}(\boldsymbol{a}_{t'}=y^{t},\boldsymbol{\theta}_{t'}=\theta) = \sum_{t'}\delta_{t'}\mathbb{P}_{p,f,\sigma}(\boldsymbol{a}_{t'}=y^{t},\boldsymbol{\theta}_{t'}=\theta).\]
Since 
\begin{align*}
\sum_{t'}\delta_{t'}\mathbb{P}_{p,f,\sigma}(\boldsymbol{a}_{t'}=y^{t},\boldsymbol{\theta}_{t'}=\theta) & =  
\sum_{t',x^{t'}}\delta_{t'}\mathbb{P}_{p,f,\sigma}(\boldsymbol{a}_{t'}=y^{t}|x^{t'},\boldsymbol{\theta}_{t'}=\theta)\mathbb{P}_{p,f,\sigma}(x^{t'},\boldsymbol{\theta}_{t'}=\theta)   \\
& = \sum_{t',x^{t'}}\delta_{t'}\mathbb{P}_{p,f,\sigma}(x^{t'},\boldsymbol{\theta}_{t'}=\theta)\sigma_{t'}(y^{t}|x^{t'}),
\end{align*}
 the result follows immediately from the observation that
 \begin{align*}
\mathbb{P}_{p,f,\sigma}(x^{t'},\boldsymbol{\theta}_{t'}=\theta) & =  \sum_{\theta_1,\dots,\theta_{t'},\dots,\theta_T : \theta_{t'}=\theta}p(\theta_1,\dots,\theta_{t'},\dots,\theta_T)f^{t'}(x^{t'}|\theta_1,\dots,\theta_{t'},\dots,\theta_T).
\end{align*}
If $T< \infty$, this completes the proof of Theorem \ref{th:evolving-states}. \medskip

If $T = \infty$, let $A = \bigcup_{t=1}^{T_n} Y^{t}$ with $T_n< +\infty$.  
The previous argument shows that there exist $\gamma_{t'}^n: X^{t'} \rightarrow \Delta(\bigcup_{t=1}^{T_n} Y^{t})$, $t'=1,\dots,+\infty$, such that 
\begin{align}
\delta_t \mathbb{P}_{p,g,\tau}(y^t,\boldsymbol{\theta}_t=\theta)  = \sum_{t'=1}^{+\infty} \delta_{t'}\mathbb{P}_{p,f,\sigma}(x^{t'},\boldsymbol{\theta}_{t'}=\theta)\gamma_{t'}^n(y^{t}|x^{t'}), \label{eq:n-suff}
\end{align}
for all $y^t$, for all $t \leq T_n$, for all $\theta$. We can extend the garbling $\gamma_{t'}^n$ into the garblings $\gamma_{t'}^n: X^{t'} \rightarrow \Delta(\bigcup_{t=1}^{+\infty} Y^{t})$ by letting $\gamma_{t'}^n(y^{t}|x^{t'})=0$ whenever $t >T_n$. 
\medskip
Therefore, the set 
\begin{align*}
\Gamma_n:=\{(\gamma_1,\dots,\gamma_{+\infty}): \text{\,Eq. (\ref{eq:n-suff}) holds for all\;} y^t, \text{\,for all\;} t \leq T_n, \text{\,for all\;} \theta\}
\end{align*}
is non-empty. It is also closed in the product topology. Finally, if $(T_n)_n$ is an increasing sequence in $n$ converging to $+ \infty$, $\Gamma_n$ is  a decreasing sequence in $n$. By Cantor's intersection theorem, $\cap_n \Gamma_n \neq \emptyset$, which completes the proof.\medskip 

Lastly, Theorem \ref{th:fixed-discounting-uncontrolled} follows as an immediate corollary by restricting attention to distributions $p$ such that $p(\theta_1,\dots,\theta_t,\dots,\theta_T)>0$ if and only if $\theta_t=\theta_t'$ for all $(t,t')$, i.e., by restricting attention to fully persistent states. 

\subsection{Proof of Theorem \ref{fixed-discounting}}
\medskip

$(1) \implies (2)$ and $(2) \implies (1)$. Immediate. 

$(3) \implies (2)$. Consider the decision problem $(A,u,\kappa_t)_t$. Fix any strategy $\tau$. Observe that 
\begin{align*}
\proba_{\theta, g, \kappa, \tau}(a_{t}|y^{t},k^{t-1})= \\
\frac{\sum_{a^{t-1}}g_1(y_1|\theta)\tau_1(a_1|y_1)\kappa_1(k_1|a_1)\times \dots \times g_t(y_t|y^{t-1},k^{t-1},\theta)\tau_{t}(a_{t}|a^{t-1},y^{t},k^{t-1})}
{\sum_{a^{t}}g_1(y_1|\theta)\tau_1(a_1|y_1)\kappa_1(k_1|a_1)\times \dots \times g_t(y_t|y^{t-1},k^{t-1},\theta)\tau_{t}(a_{t}|a^{t-1},y^{t},k^{t-1})}=
\\
\frac{\sum_{a^{t-1}}\tau_1(a_1|y_1)\kappa_1(k_1|a_1)\times \dots \times \tau_{t}(a_{t}|a^{t-1},y^{t},k^{t-1})}
{\sum_{a^{t}}\tau_1(a_1|y_1)\kappa_1(k_1|a_1)\times \dots \times \tau_{t}(a_{t}|a^{t-1},y^{t},k^{t-1})},
\end{align*}
whenever the denominator is strictly positive, i.e., $ \proba_{\theta, g, \kappa, \tau}(y^{t},k^{t-1})>0$. Note that the probability $\proba_{\theta, g, \kappa, \tau}(a_{t}|y^{t},k^{t-1})$ is independent of $\theta$. 
 \medskip

Define $\xi_{t'}: \bigcup_{t}(Y^t \times K^{t-1}) \times K^{t'-1} \rightarrow \Delta(K_{t'})$ as follows: 
\begin{align*}
\xi_{t'}(k_{t'}|(y^{t},l^{t-1}),k^{t'-1}) :=  \sum_{a}\overline{\proba}_{\theta, g, \kappa, \tau}(\boldsymbol{a_t}=a|y^{t},l^{t-1})\kappa_{t'}(k_{t'}|a,k^{t'-1}),
\end{align*}
where $\overline{ \proba}_{\theta,g,\kappa,\tau}(\cdot|y^{t},l^{t-1}) =  \proba_{\theta,g,\kappa,\tau}(\cdot|y^{t},l^{t-1})$ whenever $\proba_{\theta,g,\kappa,\tau}(y^{t},l^{t-1})>0$ and is arbitrary, otherwise.   We again stress that the kernel $\xi_t$  is independent of the parameter $\theta$.\medskip

Let $(\gamma_t)_t$ be the kernels satisfying $\delta$-sufficiency, which correspond to the kernels $(\xi_t)_t$ we have constructed. Define $\sigma_{t'}: X^{t'} \times K^{t'-1} \rightarrow \Delta(A)$ as follows: 
\begin{align*}
\sigma_{t'}(a|x^{t'},k^{t'-1}) =  \sum_{t} \sum_{y^t,l^{t-1}}\overline{ \proba}_{\theta,g,\kappa,\tau}(\boldsymbol{a_t}=a|y^{t},l^{t-1})\gamma_{t'}(y^t,l^{t-1}|x^{t'},k^{t'-1}).
\end{align*}
For future reference, notice that: 
\begin{align*}
\sum_{a}\sigma_{t'}(a|x^{t'},k^{t'-1})\kappa_{t'}(k_{t'}|a,k^{t'-1})& = \sum_{a}\sum_{t,y^t,l^{t-1}}\overline{ \proba}_{\theta,g,\kappa,\tau}(\boldsymbol{a_t}=a|y^{t},l^{t-1})\gamma_{t'}(y^t,l^{t-1}|x^{t'},k^{t'-1})\kappa_{t'}(k_{t'}|a,k^{t'-1}),  \\
& = \sum_{t,y^t,l^{t-1}} \xi_{t'}(k_{t'}|(y^{t},l^{t-1}),k^{t'-1})\gamma_{t'}(y^t,l^{t-1}|x^{t'},k^{t'-1}).
\end{align*}

We now show that 
\[\sum_{t}\delta_t\marg_{A_{t}}\proba_{\theta,g,\kappa,\tau} = \sum_{t}\delta_t\marg_{A_{t}}\proba_{\theta,f,\kappa,\sigma}.\]
We have: 
\begin{align*}
\sum_{t', x^{t'},k^{t'-1}}\delta_{t'} \proba_{\theta, f, \kappa, \sigma}(x^{t'},k^{t'-1})  \proba_{\theta, f, \kappa, \sigma}(\boldsymbol{a_{t'}}=a|x^{t'},k^{t'-1})= \\
\sum_{t', x^{t'},k^{t'-1}}\delta_{t'} \proba_{\theta, f, \kappa, \sigma}(x^{t'},k^{t'-1})  \sigma_{t'}(a|x^{t'},k^{t'-1}) = \\
\sum_{t', x^{t'},k^{t'-1}}\delta_{t'} \proba_{\theta, f, \kappa, \sigma}(x^{t'},k^{t'-1})  \left(\sum_{t,y^t,l^{t-1}}\overline{ \proba}_{\theta,g,\kappa,\tau}(\boldsymbol{a_t}=a|y^{t},l^{t-1})\gamma_{t'}(y^t,l^{t-1}|x^{t'},k^{t'-1})\right) = \\
\sum_{t,y^t,l^{t-1}} \overline{\proba}_{\theta, g, \kappa, \tau}(\boldsymbol{a_t}=a|y^{t},l^{t-1}) \left(\sum_{t', x^{t'},k^{t'-1}}\delta_{t'} \proba_{\theta, f, \kappa, \sigma}(x^{t'},k^{t'-1})\gamma_{t'}(y^t,l^{t-1}|x^{t'},k^{t'-1}) \right).
\end{align*}

We argue that $\delta$-sufficiency implies that the term in parentheses equals $\delta_t\proba_{\theta,g,\kappa,\tau}(y^{t},l^{t-1})$. By construction, we have:  
\begin{flalign*}
&  \proba_{\theta,f,\kappa,\sigma}(x^{t'},k^{{t'}-1})  =  & \\
& f_1(x_1|\theta) \times \dots \times f_{t'}(x_{t'}|x^{t'-1},k^{t'-1},\theta)\times   & \\ 
& \left(\sum_{a_1}\sigma_1(a_1|x^1)\kappa_1(k_1|a_1)\right) \times \dots \times 
 \left(\sum_{a_{t'-1}}\sigma_{t'-1}(a_{t'-1}|x^{t'-1},k^{t'-2})\kappa_{t'-1}(k_{t'-1}|a_{t'-1},k^{t'-2})\right) = & \\
& f_1(x_1|\theta) \times \dots \times f_{t'}(x_{t'}|x^{t'-1},k^{t'-1},\theta)\times & \\ 
& \left(\sum_{t,y^t,l^{t-1}} \xi_1(k_1|(y^t,l^{t-1}),k^0)\gamma_1(y^t,l^{t-1}|x^1,k^0) \right)\times \dots \times & \\
& \left(\sum_{t,y^t,l^{t-1}}\xi_{t'-1}(k_{t'-1}|(y^t,l^{t-1}),k^{t'-2})\gamma_{t'-1}(y^t,l^{t-1}|x^{t'-1},k^{t'-2})\right) = & \\ 
 & = \proba_{\theta,f, \;\xi \circ \gamma}(x^{t'},k^{{t'}-1}).&
\end{flalign*}

The condition of $\delta$-sufficiency then implies that: 
\begin{align*}
\sum_{t', x^{t'},k^{t'-1}}\delta_{t'}\proba_{\theta,f,\kappa,\sigma}(x^{t'},k^{t'-1}) \gamma_{t'}(y^t,l^{t-1}|x^{t'},k^{t'-1}) = \delta_t \proba_{\theta,g, \xi}(y^{t},l^{t-1})
\end{align*}
which is equal to $\delta_t \proba_{\theta,g,\kappa,\tau}(y^{t},l^{t-1})$ whenever $\proba_{\theta,g,\kappa,\tau}(y^{t},l^{t-1})>0$. To see this, note that 
\begin{align*}
\proba_{\theta,g,\kappa,\tau}(y^{t},l^{t-1})  = \sum_{a^{t-1}}\proba_{\theta,g,\kappa,\tau}(a^{t-1},y^{t},l^{t-1}) \\
 = \sum_{a^{t-1}}\proba_{\theta,g,\kappa,\tau}(y_{t}|a^{t-1},y^{t-1},l^{t-1})\proba_{\theta,g,\kappa,\tau}(l_{t-1}|a^{t-1},y^{t-1},l^{t-2}) \proba_{\theta,g,\kappa,\tau}(a^{t-1}|y^{t-1},l^{t-2})\proba_{\theta,g,\kappa,\tau}(y^{t-1},l^{t-2})\\
 = \sum_{a^{t-1}}g_{t}(y_{t}|y^{t-1},l^{t-1},\theta) \kappa_{t-1}(l_{t-1}|a_{t-1},l^{t-2})\proba_{\theta,g,\kappa,\tau}(a^{t-1}|y^{t-1},l^{t-2})\proba_{\theta,g,\kappa,\tau}(y^{t-1},l^{t-2})  \\
 = g_{t}(y_{t}|y^{t-1},l^{t-1},\theta) \left(\sum_{a_{t-1}}\kappa_{t-1}(l_{t-1}|a_{t-1},l^{t-2})\proba_{\theta,g,\kappa,\tau}(a_{t-1}|y^{t-1},l^{t-2})\right)\proba_{\theta,g,\kappa,\tau}(y^{t-1},l^{t-2}) \\
 =g_{t}(y_{t}|y^{t-1},l^{t-1},\theta)\xi_{t-1}(l_{t-1}|(y^{t-1},l^{t-2}), l^{t-2})\proba_{\theta,g,\kappa,\tau}(y^{t-1},l^{t-2}),
\end{align*}
so that the result follows by induction.

\medskip 

$(2) \implies (3)$. Fix any collection $(\xi_t)_t$. Let $A= \bigcup_{t}(Y^t \times K^{t-1})$ and $\kappa_t=\xi_t$ for all $t$. Consider the strategy $\tau$ such that $\tau_t(y^t,l^{t-1}|y^t,l^{t-1}, a^{t-1})=1$ for all $(a^{t-1},y^t,l^{t-1})$. According to $\proba_{\theta, g, \kappa, \tau}$, the discounted probability  of $(y^{t+1},l^{t})$  is  
\begin{align*}
\proba_{\theta, g, \kappa, \tau}(y^{t+1},l^{t}) & = 
\delta_t \left[g_1(y_1|\theta)\xi_1(l_1|y_1) \times \dots \times \xi_{t}(l_t|(y^t,l^{t-1}),l^{t-1})g_{t+1}(y_{t+1}|y^{t},l^t,\theta)\right] \\
& = \delta_t\proba_{\theta, g, \xi}(y^{t+1},l^{t})
\end{align*}
\medskip 
For $f$ to be more valuable than $g$, there must exist a strategy $\sigma$ such that 
\begin{align*}
\proba_{\theta, g, \kappa, \tau}(y^{t+1},l^{t})=\sum_{t'}\delta_{t'}\proba_{\theta, f, \kappa, \sigma}(\boldsymbol{a_{t'}}=(y^{t+1},l^{t})),
\end{align*}
for all $(y^{t+1},l^{t})$, for all $t$.

To ease notation, let $h^{t'}$ be a history of past actions of length $t'$, that is, an element of 
\[\underbrace{\bigcup_{t}(Y^t \times K^{t-1}) \times \dots \times \bigcup_{t}(Y^t \times K^{t-1})}_{t' \text{\;times}}.\] 
From the law of total probability, we have
\begin{align*}
\proba_{\theta, f, \kappa, \sigma}(y^{t+1},l^{t})= \\
\sum_{(x^{t'+1},k^{t'})}\sum_{h^{t'}} \proba_{\theta, f, \kappa, \sigma}\left(y^{t+1},l^{t}| x^{t'+1},k^{t'}, h^{t'}\right)\proba_{\theta, f, \kappa, \sigma}\left(h^{t'} |x^{t'+1},k^{t'}\right)\proba_{\theta, f, \kappa, \sigma}(x^{t'+1},k^{t'}) = \\
\sum_{(x^{t'+1},k^{t'})}\sum_{h^{t'}} \sigma_{t+1}\left(y^{t+1},l^{t}| x^{t'+1},k^{t'}, h^{t'}\right)\proba_{\theta, f, \kappa, \sigma}\left(h^{t'} |x^{t'+1},k^{t'}\right)\proba_{\theta, f, \kappa, \sigma}(x^{t'+1},k^{t'}) =\\
\sum_{(x^{t'+1},k^{t'})}\gamma_{t'+1}\left(y^{t+1},l^{t}|x^{t'+1},k^{t'}\right) \proba_{\theta, f, \kappa, \sigma}(x^{t'+1},k^{t'}) ,
\end{align*}
where we define 
\begin{align*}
\gamma_{t'+1}\left(y^{t+1},l^{t}|x^{t'+1},k^{t'}\right) := \sum_{h^{t'}} \sigma_{t+1}\left(y^{t+1},l^{t}| x^{t'+1},k^{t'}, h^{t'}\right)\proba_{\theta, f, \kappa, \sigma}\left(h^{t'} |x^{t'+1},k^{t'}\right).
\end{align*}
Note that $\gamma_{t+1}: X^{t+1} \times K^{t} \rightarrow \Delta(\bigcup_{t}(Y^t \times K^{t-1}))$ is a well-defined kernel and independent of $\theta$, since $\proba_{\theta, f, \kappa, \sigma}\left(h^{t'} |x^{t'+1},k^{t'}\right)$ is independent of $\theta$. We now compute $\proba_{\theta, f, \kappa, \sigma}(x^{t'+1},k^{t'})$ as follows: 
\begin{align*}
\proba_{\theta, f, \kappa, \sigma}(x^{t'+1},k^{t'})& = \proba_{\theta, f, \kappa, \sigma}(x_{t'+1}|x^{t'},k^{t'})\proba_{\theta, f, \kappa, \sigma}(x^{t'},k^{t'})\\ & = f_{t'}(x_{t'+1}|x^{t'},k^{t'},\theta)\proba_{\theta, f, \kappa, \sigma}(x^{t'},k^{t'}),
\end{align*}
and 
\begin{align*}
\proba_{\theta, f, \kappa, \sigma}(x^{t'},k^{t'}) = \\
\sum_{h^{t'}, \tilde{y}^{t},\tilde{l}^{t-1}}\biggl(\proba_{\theta, f, \kappa, \sigma}(k_{t'}|\tilde{y}^{t},\tilde{l}^{t-1},x^{t'},k^{t'-1}, h^{t'})\proba_{\theta, f, \kappa, \sigma}(\tilde{y}^{t},\tilde{l}^{t-1}|x^{t'},k^{t'-1}, h^{t'})\times \\
\proba_{\theta, f, \kappa, \sigma}(h^{t'}|x^{t'},k^{t'-1})\proba_{\theta, f, \kappa, \sigma}(x^{t'},k^{t'-1}) \biggl)=\\
\left(\sum_{\tilde{y}^{t},\tilde{l}^{t-1}}\xi_{t'}(k_{t'}|(\tilde{y}^{t},\tilde{l}^{t-1}),k^{t'-1})\gamma_{t'}(\tilde{y}^t,\tilde{l}^{t-1}|x^{t'},k^{t'-1})\right)\proba_{\theta, f, \kappa, \sigma}(x^{t'},k^{t'-1}),
\end{align*}
where we use the identities: 
\begin{align*}
 \proba_{\theta, f, \kappa, \sigma}(k_{t'}|\tilde{y}^{t},\tilde{l}^{t-1},x^{t'},k^{t'-1}, h^{t'}) = \xi_t(k_{t'}|(\tilde{y}^{t},\tilde{l}^{t-1}), k^{t'-1}) \\
\gamma_{t'}(\tilde{y}^t,\tilde{l}^{t-1}|x^{t'},k^{t'-1})  = \sum_{h^{t'}} \proba_{\theta, f, \kappa, \sigma}(\tilde{y}^{t},\tilde{l}^{t-1}|x^{t},k^{t-1}, h^{t'})\proba_{\theta, f, \kappa, \sigma}(h^{t'}|x^{t'},k^{t'-1}).
\end{align*}
Therefore, 
\begin{align*}
\proba_{\theta, f, \kappa, \sigma}(x^{t'+1},k^{t'}) = \\
 f_t(x_{t'+1}|x^{t'},k^{t'},\theta)\left(\sum_{\tilde{y}^{t},\tilde{l}^{t-1}}\xi_{t'}(k_{t'}|(\tilde{y}^{t},\tilde{l}^{t-1}),k^{t'-1})\gamma_{t'}(\tilde{y}^{t},\tilde{l}^{t-1}|x^{t'},k^{t'-1})\right)\proba_{\theta, f, \kappa, \sigma}(x^{t'},k^{t'-1}).
\end{align*}
Iterating the argument,  it follows that  $\proba_{\theta, f, \kappa, \sigma}(x^{t'+1},k^{t'}) =  \proba_{\theta, f, \xi \circ \gamma}(x^{t'+1},k^{t'})$, which completes the proof if $T < + \infty$.

Finally, if $T=+\infty$, an argument similar to the one of the proof of Theorem \ref{th:fixed-discounting-uncontrolled} completes the proof.

\subsection{Proof of Theorem \ref{th:separable-uncontrolled}}
The proof of the equivalence between (1), (2) and (3) is nearly identical to the proof of Theorem \ref{th:fixed-discounting-uncontrolled} and left to the reader. 

\section{Constrained Decisions}

 Suppose that the decision-maker is constrained with $\mathcal{A} \subseteq A_1 \times \dots \times A_t \times \dots \times A_{T}$ the feasible sequences of decisions. Let $\mathcal{A}^t$ be the projection of $\mathcal{A}$ on its first $t$ coordinates. The strategy $\tau$ is feasible if $\tau_t(a_t|a^{t-1},y^t) >0$ only if $(a^{t-1},a_t) \in \mathcal{A}^t$. Similarly, for $\sigma$. We have the following theorem.

\begin{theorem}\label{th:constrained choices}
The experiment $f$ is more valuable than the experiment $g$ in all constrained discounted problems  if, and only if, $f$ is more valuable than $g$ in all constrained sequential problems.
\end{theorem}

\begin{proof}[ of Theorem \ref{th:constrained choices}] $(\Leftarrow)$. Immediate. \medskip 

$(\Rightarrow)$. Assume $T < \infty$. (Again, if $T=\infty$, we can replicate the argument made in the proof of Theorem \ref{th:fixed-discounting-uncontrolled}.) Let $A =  Y^1 \cup \dots \cup Y^t \cup \dots \cup Y^T$ be the set of decisions. The set of feasible sequences is defined as follows:   $(a_1,\dots,a_t,\dots, a_T) \in \mathcal{A}$ if, and only if, 
$a_t \in Y^t$ and the projection of $a_t$ on its first $t' \leq t$ coordinates is $a_{t'}$, for all $(t,t')$. In words, a sequence of decisions is feasible if each new decision completes the previous one. \medskip 

Let $\tau$ be such that $\tau_t(a_t|y^t,a^{t-1})=1$ whenever $a_t=y^t$ for all ``on-path'' histories $(y^t,a^{t-1})$. (A history $(y^t,a^{t-1})$ is on-path if the projection of $y^t$ on its first $(t-1)$ coordinates is $a^{t-1}$.) The strategy is unspecified, otherwise.  Hence, $\tau_1(y_1|y_1) =1$ for all $y_1$, $\tau_2((y_1,y_2) |(y_1,y_2),y_1) =1$ for all $(y_1,y_2)$, etc. \medskip 

Under $\tau$, the discounted probability of $y^t$ is $\delta_t g^t(y^t|\theta)$ when the state is $\theta$. \medskip 

By construction, the action $y^t$ can only be chosen at period $t$. Hence, to replicate the discounted probability of $y^t$, there must exist a feasible strategy $\sigma$ such that  
\begin{align*}
\delta_t g^t(y^t|\theta) = \delta_t \sum_{x^t,\tilde{y}^1,\dots,\tilde{y}^{t-1}} f^t(x^t|\theta) \sigma_1(\tilde{y}^1|x^1)\sigma_2(\tilde{y}^2|x^2,\tilde{y}_1) \times \dots \times \sigma_t(y^t|x^t,(\tilde{y}^1,\dots,\tilde{y}^{t-1})). 
\end{align*}
Since feasibility requires that $\sigma_{t'}(\hat{y}^{t'}|x^{t'}, (\tilde{y}^1,\dots, \tilde{y}^{t'-1}))>0$  only if the projection of $\hat{y}^{t'}$ on any of its first $t^{''} \leq t'$ coordinates is $\tilde{y}^{t^{''}}$, the above is equivalent to:
\begin{align*}
\delta_t g^t(y^t|\theta) = \delta_t \sum_{x^t} f^t(x^t|\theta) \sigma_1(y^1|x^1)\sigma_2(y_2|x^2,y^1) \times \dots \times \sigma_t(y_t|x^t,(y^1,\dots,y^{t-1})). 
\end{align*}
Letting $\gamma_{t}(y_t|x^t,y^{t-1}) := \sigma_t(y_t|x^t,(y^1,\dots,y^{t-1}))$, we can rewrite the above expression as:
\begin{align*}
g^t(y^t|\theta) = \sum_{x^t} f^t(x^t|\theta) \gamma_1(y^1|x^1)\gamma_2(y_2|x^2,y^1) \times \dots \times \gamma_t(y_t|x^t,y^{t-1}), 
\end{align*}
which is the sufficiency condition for sequential problems --  the garbling is adapted. 
\end{proof}

Remark: In the simpler case, where $\mathcal{A} = A_1 \times \dots \times A_t \times \dots \times A_T$, with $A_t \subseteq A$ for all $t$, the same construction (i.e., using $A = \cup_t Y^t$ and $A_t=Y^t$) gives the equivalence with separable problems. In this case, we define the garbling $\gamma_t$ as: 
\begin{align*}
\gamma_t(y^t|x^t) := \sum_{\tilde{y}^1,\dots,\tilde{y}^{t-1}}\sigma_1(\tilde{y}^1|x^1)\sigma_2(\tilde{y}^2|x^2,y_1) \times \dots \times \sigma_t(y^t|x^t,(\tilde{y}^1,\dots,\tilde{y}^{t-1})),
\end{align*}
to obtain 
\begin{align*}
g^t(y^t|\theta) = \sum_{x^t}f^t(x^t|\theta)\gamma_t(y^t|x^t).
\end{align*}

\section{Online Appendix: Theorem \ref{th:comparison-bernouilli} and Additional Materials} 
This section contains the proof  of Theorem \ref{th:comparison-bernouilli} along with some additional materials on the comparison of Bernouilli distributions. We start with the proof of Theorem \ref{th:comparison-bernouilli}. \medskip

Recall that the experiment $f$ induces the following posteriors (the table also includes the probability of each posterior):  

\begin{table}[h!]
\resizebox{\columnwidth}{!}{%
\begin{tabular}{|c|c|c|c|c|c|c|}\hline 
signal: & $0$ & $1$ & $00$ & $01$ & $10$ & $11$ \\ \hline
$[\theta=0]$ & $\delta_1 p$ & $\delta_1(1-p)$ & $\delta_2p q$ & $\delta_2p(1-q)$ & $\delta_2(1-p)q$ & $\delta_2(1-p)(1-q)$  \\ \hline
$[\theta=1]$ & $\delta_1(1-p)$ & $\delta_1p$ &  $\delta_2(1-p)(1-q)$ & $\delta_2(1-p)q$ & $\delta_2 p(1-q)$ & $\delta_2 p q$ \\ \hline 
$2 \times $ proba. signal: & $\delta_1$ & $\delta_1$ & $\delta_2(p q + (1-p)(1-q))$ & $\delta_2(p(1-q) + (1-p)q)$ & $\delta_2((1-p)q+p(1-q))$ & $\delta_2 ((1-p)(1-q) + pq)$ \\ \hline
posterior of $[\theta=0]$: & $p$ & $1-p $ & $\frac{pq}{p q + (1-p)(1-q)} $ & $\frac{p(1-q)}{p(1-q) + (1-p)q}$ & $\frac{(1-p)q}{p(1-q) + (1-p)q}$   & $\frac{(1-p)(1-q)}{p q + (1-p)(1-q)}$ \\ \hline
\end{tabular}}
\end{table} 

To ease the exposition, we adopt the following notation: 
\begin{align*}
\pi_{11} & := \frac{(1-p)(1-q)}{p q + (1-p)(1-q)} \leq 1/2, \\
\pi_{1}  &:= 1-p  \leq  1/2,\\
\pi_{10} & := \frac{(1-p)q}{p(1-q) + (1-p)q}\leq 1/2, \\ 
\lambda  & := p q + (1-p)(1-q).
\end{align*} 
Note that $\pi_{11} \leq \pi_{1} \leq \pi_{10} \leq  1/2 \leq 1-\pi_{10} = \pi_{01} \leq  1-\pi_{1}= \pi_{0} \leq 1-\pi_{11}= \pi_{00}$, with strict inequalities if $p > q >1/2$. ($p \geq q$ is used to rank $\pi_{10}$ and $1-\pi_{10}$.) For future reference, observe that $ \lambda \pi_{11} + (1-\lambda)\pi_{10}=\pi_{1}$. 

Following Theorem 12.4.1 of \cite{blackwell-girshick}, the mixture experiment $\delta_1 f^1 + \delta_2 f^2$ is sufficient for the mixture experiment $\delta_1 g^1 + \delta_2g^2$ if, and only if, the cumulative distribution $F$ over posteriors induced by $f$ is a mean-preserving spread of the cumulative distribution $G$ over posteriors induced by $g$.\medskip 

The distribution  $F$ is  the increasing step function:  
\begin{align*}
2 \times F(\pi) = 
\begin{cases}
 \delta_2 \lambda  & \pi = \pi_{11}\\
\delta_2 \lambda + \delta_1  & \pi =\pi_{1} \\
1 & \pi = \pi_{10} \\
  1 + \delta_2(1-\lambda) & \pi = 1-\pi_{10} \\
1 + \delta_2(1-\lambda) + \delta_1 & \pi =1-\pi_{1} \\
 2 & \pi=1-\pi_{11}
\end{cases}
\end{align*}
We now compute $2\int_0^t F(\pi)d\pi$ for all $t \in [0,1]$: 
\begin{align*}
2\int_0^t F(\pi)d\pi = 
\begin{cases}
0 & t <\pi_{11} \\
t \delta_2 \lambda - \pi_{11} \delta_2\lambda & \pi_{11} \leq t <\pi_{1} \\
t (\delta_2 \lambda +\delta_1) - \pi_{11}\delta_2 \lambda -\pi_{1} \delta_1& \pi_{1} \leq t <\pi_{10} \\
t -  \pi_{1}       & \pi_{10} \leq t <1-\pi_{10} \\ 
t (2-\delta_2\lambda-\delta_1) + (1-\pi_{11})\delta_2 \lambda + (1-\pi_{1})\delta_1-1 &  1-\pi_{10} \leq t <1-\pi_{1}\\
t (2-\delta_2\lambda) + (1-\pi_{11})\delta_2 \lambda -1 & 1-\pi_{1} \leq t < 1-\pi_{11} \\ 
2t-1 &  1-\pi_{11} \leq t \leq 1 \\
\end{cases}
\end{align*}
(To compute the above, we have repeatedly used the observation that $\delta_2 \lambda +\delta_2(1-\lambda) + \delta_1 =1$.) We call the later function $H : [0,1] \rightarrow \mathbb{R}_+$. It is easy to check that the function is piece-wise linear, increasing and convex, thus it is the maximum of all the linear functions defining it, i.e., for all $t \in [0,1]$,
\begin{align*}
H(t) = \max \Big(0,&\; t \delta_2 \lambda - \pi_{11} \delta_2\lambda,  t (\delta_2 \lambda +\delta_1) - \pi_{11}\delta_2 \lambda -\pi_{1} \delta_1, t -  \pi_{1}, \\ 
& t (2-\delta_2\lambda-\delta_1) + (1-\pi_{11})\delta_2 \lambda + (1-\pi_{1})\delta_1-1, t (2-\delta_2\lambda) + (1-\pi_{11})\delta_2 \lambda -1, 2t-1\Big).
\end{align*}
\medskip

Similarly, the mixture experiment $\delta_1 g^1 + \delta_2g^2$, parameterized by  $(p',q')$, induces the function $H'$, with $\pi_{11}',\pi_{1}',\pi_{10}',\lambda'$ instead of $\pi_{11},\pi_{1},\pi_{10},\lambda$.\medskip 

To summarize, the experiment $f$ is $\delta$-sufficient for $g$  if, and only if, $H(t) \geq H'(t)$, for all $t$.\medskip

A necessary and sufficient condition for $H \geq H'$ is that $H(t) \geq H'(t)$ at $t=\pi_{11},\pi_{1}, \pi_{10},1-\pi_{10},1-\pi_{1},1-\pi_{11}$. To see this, observe that if $H(\pi_{11}) \geq H'(\pi_{11})$ and $H'(\pi_{1}) \geq H'(\pi_{1})$, then 
\begin{align*}
H(\beta \pi_{11} + (1-\beta) \pi_{1}) & = \beta H(\pi_{11})+ (1-\beta)H(\pi_{1}) \\
& \geq \beta H'(\pi_{11})+ (1-\beta)H'(\pi_{1}) 
& \geq H'(\beta \pi_{11} + (1-\beta) \pi_{1}),
\end{align*}
for all $\beta \in [0,1]$, where the first equality follows from the linearity of $H$ in the segment $[\pi_{11},\pi_{1}]$ and the last inequality follows from the convexity of $H'$. Repeating the argument segment by segment proves the statement. We now compute these six inequalities.\medskip 

\underline{\textsc{(a): $H(\pi_{11}) \geq H'(\pi_{11})$.}}

\begin{align*}
H(\pi_{11}) =0 \notag \\
 \geq 0  \tag{1}\\
\geq (\pi_{11}-\pi_{11}')\delta_2 \lambda'  \tag{2}\\
\geq \pi_{11} (\delta_2 \lambda' +\delta_1) -\pi_{11}'\delta_2 \lambda' -\pi_{1}' \delta_1 \tag{3}\\
\geq \pi_{11}-\pi_{1}' \tag{4}\\
\geq \pi_{11}(2-\delta_2 \lambda' -\delta_1) + (1-\pi_{11}')\delta_2\lambda' + (1-\pi_{1}')\delta_1 -1 \tag{5}\\
\geq \pi_{11}(2-\delta_2 \lambda' ) + (1-\pi_{11}')\delta_2\lambda'  -1\tag{6}\\
\geq 2 \pi_{11}-1 \tag{7}
\end{align*}

We can verify that if (2) is satisfied, so is (6). It suffices to show that 
\begin{align*}
(\pi_{11}-\pi_{11}')\delta_2 \lambda'  \geq \pi_{11}(2-\delta_2 \lambda' ) + (1-\pi_{11}')\delta_2x'  -1 \Leftrightarrow 2 \pi_{11}(\delta_2 \lambda' -1) \geq \delta_2x'-1,
\end{align*}
which is true because $\pi_{11} \leq 1/2$ and $\delta_2 \lambda' \leq 1$. The same argument shows that if (3) is satisfied, so is (5). (We repeatedly use such a pattern, which follows from the symmetry of the experiments.) Also, (7) is satisfied because $\pi_{11} \leq 1/2$. Finally, rewriting (3) as 
\begin{align*}
(\pi_{11} -\pi_{11}') \delta_2x' + (\pi_{11}-\pi_{1}') \delta_1,
\end{align*} 
it follows that if (2) and (4) are satisfied, so is (3). Therefore, we have the following. 

\begin{lemma}
$H(\pi_{11}) \geq H'(\pi_{11})$ if, only if,  
\begin{align*}
0 &\geq (\pi_{11}-\pi_{11}')\delta_2 \lambda' \\ 
0 & \geq \pi_{11}-\pi_{1}'
\end{align*} 
\end{lemma}
\medskip

\underline{\textsc{(b): $H(\pi_{1}) \geq H'(\pi_{1})$.}}

\begin{align}
H(\pi_{11}) =(\pi_{1}-\pi_{11})\delta_2 \lambda \notag \\
 \geq 0 \tag{1}\\
\geq (\pi_{1}-\pi_{11}')\delta_2 \lambda' \tag{2}\\
\geq \pi_{1} (\delta_2 \lambda' +\delta_1) -\pi_{11}'\delta_2 \lambda' -\pi_{1}' \delta_1 \tag{3}\\
\geq \pi_{1}-\pi_{1}' \tag{4}\\
\geq \pi_{1}(2-\delta_2 \lambda' -\delta_1) + (1-\pi_{11}')\delta_2\lambda' + (1-\pi_{1}')\delta_1 -1 \tag{5}\\
\geq \pi_{1}(2-\delta_2 \lambda' ) + (1-\pi_{11}')\delta_2\lambda'  -1\tag{6}\\
\geq 2 \pi_1-1 \tag{7}
\end{align}
We show that if (2) holds, so does (6). It suffices to show that 
\begin{align*}
 (\pi_{1}-\pi_{11}')\delta_2 \lambda' \geq  \pi_{1}(2-\delta_2 \lambda' ) + (1-\pi_{11}')\delta_2\lambda'  -1 \Leftrightarrow 2 \pi_{1}(\delta_2\lambda'-1) \geq \delta_2\lambda'-1,
\end{align*}
which holds since $\pi_{1} \leq 1/2$ and $\delta_2\lambda'\leq 1$. Similarly, (3) implies (5). (1) also clearly holds. 

\begin{lemma}
$H(\pi_{1}) \geq H'(\pi_{1})$ if, and only if, 
\begin{align*}
(\pi_{1}-\pi_{11})\delta_2 \lambda & \geq (\pi_{1}-\pi_{11}')\delta_2 \lambda' \\
(\pi_{1}-\pi_{11})\delta_2\lambda & \geq (\pi_{1}-\pi_{11}')\delta_2\lambda' + (\pi_{1}-\pi_{1}')\delta_1 \\
(\pi_{1}-\pi_{11}) \delta_2\lambda  & \geq \pi_{1}-\pi_{1}'
\end{align*}
\end{lemma}
\medskip

\underline{\textsc{(c): $H(\pi_{10}) \geq H'(\pi_{10})$.}}

\begin{align}
H(\pi_{10}) =(\pi_{10}-\pi_{1}) \notag \\
 \geq 0 \tag{1}\\
\geq (\pi_{10}-\pi_{11}')\delta_2 \lambda' \tag{2}\\
\geq \pi_{10} (\delta_2 \lambda' +\delta_1) -\pi_{11}'\delta_2 \lambda' -\pi_{1}' \delta_1 \tag{3}\\
\geq \pi_{10}-\pi_{1}' \tag{4}\\
\geq \pi_{10}(2-\delta_2 \lambda' -\delta_1) + (1-\pi_{11}')\delta_2\lambda' + (1-\pi_{1}')\delta_1 -1 \tag{5}\\
\geq \pi_{10}(2-\delta_2 \lambda' ) + (1-\pi_{11}')\delta_2\lambda'  -1\tag{6}\\
\geq 2 \pi_{10}-1 \tag{7}
\end{align}
We show that if (2) holds, so does (6). It suffices to show that 
\begin{align*}
(\pi_{10}-\pi_{11}')\delta_2 \lambda' \geq \pi_{10}(2-\delta_2 \lambda' ) + (1-\pi_{11}')\delta_2\lambda'  -1 \Leftrightarrow 2 \pi_{10}(\delta_2 \lambda' -1) \geq (\delta_2 \lambda'-1),
\end{align*}
which holds since $\pi_{1} \leq 1/2$ and $\delta_2\lambda'\leq 1$. Similarly, (3) implies (5). (1) also clearly holds. 

\begin{lemma}
$H(\pi_{10}) \geq H'(\pi_{10})$ if, only if, 
\begin{align*}
(\pi_{10}-\pi_{1})  & \geq (\pi_{10}-\pi_{11}')\delta_2 \lambda' \\
(\pi_{10}-\pi_{1}) & \geq (\pi_{10}-\pi_{11}')\delta_2\lambda' + (\pi_{10}-\pi_{1}')\delta_1 \\
(\pi_{10}-\pi_{1}) & \geq (\pi_{10}-\pi_{1}')
\end{align*}
\end{lemma}

\medskip

\underline{\textsc{(d): $H(1-\pi_{10}) \geq H'(1-\pi_{10})$.}}

\begin{align}
H(1-\pi_{10}) =1-\pi_{10}-\pi_{1} \notag \\
 \geq 0 \tag{1}\\
\geq (1-\pi_{10}-\pi_{11}')\delta_2 \lambda' \tag{2}\\
\geq (1-\pi_{10}) (\delta_2 \lambda' +\delta_1) -\pi_{11}'\delta_2 \lambda' -\pi_{1}' \delta_1 \tag{3}\\
\geq 1-\pi_{10}-\pi_{1}' \tag{4}\\
\geq (1-\pi_{10})(2-\delta_2 \lambda' -\delta_1) + (1-\pi_{11}')\delta_2\lambda' + (1-\pi_{1}')\delta_1 -1 \tag{5}\\
\geq (1-\pi_{10})(2-\delta_2 \lambda' ) + (1-\pi_{11}')\delta_2 \lambda'  -1\tag{6}\\
\geq 2(1-\pi_{10})-1 \tag{7}
\end{align}
We show that if (6) holds, so does (2). It suffices to show that
\begin{align*}
(1-\pi_{10})(2-\delta_2 \lambda' ) + (1-\pi_{11}')\delta_2\lambda'  -1 \geq (1-\pi_{10}-\pi_{11}')\delta_2 \lambda' \Leftrightarrow 2(1-\pi_{10})(1-\delta_2\lambda') \geq 1-\delta_2\lambda',
\end{align*}
which holds because $1-\pi_{10} \geq 1/2$ and $1 \geq \delta_2\lambda'$. Similarly, (5) implies (3). Clearly, (1) and (7) also hold.  
\medskip 

Consequently, 
\begin{lemma}
$H(1-\pi_{10}) \geq H'(1-\pi_{10})$ if, only if, 
\begin{align*}
(\pi_{10}-\pi_{1})  & \geq (\pi_{10}-\pi_{11}')\delta_2 \lambda' \\
(\pi_{10}-\pi_{1}) & \geq (\pi_{10}-\pi_{11}')\delta_2\lambda' + (\pi_{10}-\pi_{1}')\delta_1 \\
(\pi_{10}-\pi_{1}) & \geq (\pi_{10}-\pi_{1}')
\end{align*}
\end{lemma}
It is worth noting that the conditions are the same as for $H(\pi_{10}) > H'(\pi_{10})$. Again, this follows from the symmetry of the problem. \medskip

\underline{\textsc{(e): $H(1-\pi_{1}) \geq H'(1-\pi_{1})$.}}

\begin{align}
H(1-\pi_{1}) =(1-\pi_{1})(2-\delta_2 \lambda) + (1-\pi_{11})\delta_2\lambda -1 \notag \\
 \geq 0 \tag{1}\\
\geq (1-\pi_{1}-\pi_{11}')\delta_2 \lambda' \tag{2}\\
\geq (1-\pi_{1}) (\delta_2 \lambda' +\delta_1) -\pi_{11}'\delta_2 \lambda' -\pi_{1}' \delta_1 \tag{3}\\
\geq 1-\pi_{1}-\pi_{1}' \tag{4}\\
\geq (1-\pi_{1})(2-\delta_2 \lambda' -\delta_1) + (1-\pi_{11}')\delta_2\lambda' + (1-\pi_{1}')\delta_1 -1 \tag{5}\\
\geq (1-\pi_{1})(2-\delta_2 \lambda' ) + (1-\pi_{11}')\delta_2\lambda'  -1\tag{6}\\
\geq 2(1-\pi_{1})-1 \tag{7}
\end{align}
We show that if (6) holds, so does (2). It suffices to show that
\begin{align*}
(1-\pi_{1})(2-\delta_2 \lambda' ) + (1-\pi_{11}')\delta_2 \lambda'  -1 \geq (1-\pi_{1}-\pi_{11}')\delta_2 \lambda' \Leftrightarrow 2(1-\pi_{1})(1-\delta_2\lambda') \geq 1-\delta_2\lambda',
\end{align*}
which holds because $1-\pi_{1} \geq 1/2$ and $1 \geq \delta_2 \lambda'$. Similarly, (5) implies (3). Clearly, (1) and (7) also hold.  Consequently, 

\begin{lemma}
$H(1-\pi_{1}) \geq H'(1-\pi_{1})$ if, only if, 
\begin{align*}
(\pi_{1}-\pi_{11})\delta_2 \lambda & \geq (\pi_{1}-\pi_{11}')\delta_2 \lambda' \\
(\pi_{1}-\pi_{11})\delta_2 \lambda & \geq (\pi_{1}-\pi_{11}')\delta_2\lambda' + (\pi_{1}-\pi_{1}')\delta_1 \\
(\pi_{1}-\pi_{11}) \delta_2 \lambda & \geq \pi_{1}-\pi_{1}'
\end{align*}
\end{lemma}
\medskip

\underline{\textsc{(f): $H(1-\pi_{11}) \geq H'(1-\pi_{11})$.}}

\begin{align}
H(1-\pi_{11}) =2(1-\pi_{11})-1 \notag \\
 \geq 0 \tag{1}\\
\geq (1-\pi_{11}-\pi_{11}')\delta_2 \lambda' \tag{2}\\
\geq (1-\pi_{11}) (\delta_2 \lambda' +\delta_1) -\pi_{11}'\delta_2 \lambda' -\pi_{1}' \delta_1 \tag{3}\\
\geq 1-\pi_{11}-\pi_{1}' \tag{4}\\
\geq (1-\pi_{11})(2-\delta_2 \lambda' -\delta_1) + (1-\pi_{11}')\delta_2\lambda' + (1-\pi_{1}')\delta_1 -1 \tag{5}\\
\geq (1-\pi_{11})(2-\delta_2 \lambda' ) + (1-\pi_{11}')\delta_2\lambda'  -1\tag{6}\\
\geq 2(1-\pi_{11})-1 \tag{7}
\end{align}
We show that if (6) holds, so does (2). It suffices to show that
\begin{align*}
 (1-\pi_{11})(2-\delta_2 \lambda' ) + (1-\pi_{11}')\delta_2\lambda'  -1 \geq  (1-\pi_{11}-\pi_{11}')\delta_2 \lambda' \Leftrightarrow 2(1-\pi_{11})(1-\delta_2\lambda') \geq 1-\delta_2\lambda',
\end{align*}
which holds because $1-\pi_{11} \geq 1/2$ and $1 \geq \delta_2\lambda'$. Similarly, (5) implies (3). Clearly, (1) and (7) also hold. Consequently, 
\begin{lemma}
$H(1-\pi_{11}) \geq H'(1-\pi_{11})$ if, only if,   
\begin{align*}
0 & \geq (\pi_{11}-\pi_{11}')\delta_2 \lambda' \\ 
0 & \geq \pi_{11}-\pi_{1}'
\end{align*} 
\end{lemma}
\medskip 

This completes the computation of the six relevant inequalities. Eliminating redundant inequalities (e.g., $\pi_1' \geq \pi_1$ implies that $\pi_1' \geq \pi_{11}$), we obtain the following proposition:

\begin{proposition}\label{prop:sufficiency}
The experiment $f$ is $\delta$-sufficient for $g$ if, and only if,
\begin{align*}
0 & \geq (\pi_{11}-\pi_{11}')\delta_2 \lambda'  \tag{1}\\
(\pi_{1}-\pi_{11})\delta_2\lambda & \geq (\pi_{1}-\pi_{11}')\delta_2 \lambda' \tag{2}\\
(\pi_{10}-\pi_{1})  & \geq (\pi_{10}-\pi_{11}')\delta_2 \lambda' \Leftrightarrow (\pi_{10}-\pi_{1})(1-\delta_2 \lambda')   \geq (\pi_{1}-\pi_{11}')\delta_2 \lambda' \tag{3} \\
(\pi_{10}-\pi_{1}) & \geq (\pi_{10}-\pi_{11}')\delta_2\lambda' + (\pi_{10}-\pi_{1}')\delta_1 \Leftrightarrow  (\pi_{10}-\pi_{1}) \delta_2(1- \lambda')   \geq (\pi_{1}-\pi_{11}')\delta_2 \lambda' + (\pi_{1}-\pi_{1}')\delta_1 \tag{4}\\
- \pi_{1}  & \geq - \pi_{1}' \tag{5}
\end{align*} 
\end{proposition}

It follows from Proposition \ref{prop:sufficiency} that $\pi_1 \leq \pi_1'$ is necessary, regardless of $\delta_2$, and also sufficient when $\delta_2=0$. In what follows, we assume that $\delta_2>0$, so that $\pi_{11} \leq \pi'_{11}$ (by construction, $\lambda' >0$). \medskip  

We start with the following claim:  $\pi_{10} \geq \pi_{1'}$. To the contrary, assume that $\pi_{10} < \pi_{1}'$. This is equivalent to $\frac{\pi_{1}-\lambda \pi_{11}}{1-\lambda} < \pi_{1}'$, i.e., $\pi_{1} -\pi_{1}' < \lambda(\pi_{11}-\pi_{1}')$. This is impossible because 
 $\pi_{1}-\pi_{1}' \geq \pi_{11} -\pi_{1}'$ and $\lambda  \leq 1$.

 Since $\pi_{10} \geq \pi_{1}'$, it follows that (4) implies (3). To see this, note that if (4) holds, i.e., 
 \begin{align*}
 (\pi_{10}-\pi_{1}) \delta_2(1- \lambda')  & \geq (\pi_{1}-\pi_{11}')\delta_2 \lambda' + (\pi_{1}-\pi_{1}')\delta_1,
 \end{align*}
 then 
 \begin{align*}
 (\pi_{10}-\pi_{1})(1-\delta_2 \lambda') & =(\pi_{10}-\pi_{1}) \delta_2(1- \lambda')  + (\pi_{10}-\pi_{1})\delta_1 \\ 
 & \geq (\pi_{1}-\pi_{11}')\delta_2 \lambda' + (\pi_{1}-\pi_{1}')\delta_1+  (\pi_{10}-\pi_{1})\delta_1\\
  & = (\pi_{1}-\pi_{11}')\delta_2 \lambda' + (\pi_{10}-\pi_{1}')\delta_1 \\
  & \geq (\pi_{1}-\pi_{11}')\delta_2 \lambda',
 \end{align*}
 that is, (3) holds.

 \underline{\textsc{Case 1}.} If $\pi_{1} \leq \pi_{11}'$, then all inequalities (1)-(5) are trivially satisfied, since the LHS are negative, while the RHS positive.\medskip

 \underline{\textsc{Case 2.}} If $\pi_{1}' \geq  \pi_{1} > \pi_{11}' \geq \pi_{11}$. We only need to satisfy (2) and (4), i.e., 
  \begin{align*}
(\pi_{1}-\pi_{11})\delta_2 \lambda & \geq (\pi_{1}-\pi_{11}')\delta_2 \lambda' \tag{2}\\
\frac{\lambda}{1-\lambda} (\pi_{1}-\pi_{11}) \delta_2(1- \lambda')   & \geq (\pi_{1}-\pi_{11}')\delta_2 \lambda' + (\pi_{1}-\pi_{1}')\delta_1 \tag{4}  
  \end{align*}
 
 Suppose that 
 \[\frac{\lambda}{1-\lambda}(\pi_{1} - \pi_{11}) (1-\lambda') \geq (\pi_1-\pi_{11}')\lambda'.\]
 Clearly, (4) is satisfied since $\pi_1 \leq \pi_{1}'$. Moreover, (3) is also satisfied if $\lambda \geq \lambda'$ since $\pi_1 - \pi_{11} \geq \pi_1 -\pi_{11'}$. If $\lambda' > \lambda$, we have 
\begin{align*}
 (\pi_{1}-\pi_{11})\lambda \geq \underbrace{\frac{1-\lambda}{1-\lambda'}}_{>1}(\pi_{1}-\pi_{11}') \lambda' \geq (\pi_{1}-\pi_{11}') \lambda',
 \end{align*}
i.e., (3) is satisfied. Thus, the only remaining case is when 
\[\frac{\lambda}{1-\lambda}(\pi_{1} - \pi_{11}) (1-\lambda') < (\pi_1-\pi_{11}')\lambda',\]
 in which case we must impose (3) and (4). 
 
 Finally, it is immediate to verify that the condition  \[\frac{\lambda}{1-\lambda}(\pi_{1} - \pi_{11}) (1-\lambda') \geq (\pi_1-\pi_{11}')\lambda'\] is identical to the one in the main text. This completes the proof.

 \subsection{Comparison of The Second-period Experiment (Viewed as A Static Experiment)}
 Suppose that we want to compare the static experiment $f^2$ and $g^2$, i.e., two independent draws from Bernouilli distributions with probabilities $(p,q)$ and $(p',q')$, respectively. (Equivalently, $\delta_2=1$.) Repeating the same arguments as above (with no $\pi_{1}$), we obtain that 
 
 \begin{lemma}The experiment $f^2$ is sufficient for $g^2$ if, and only if, 
 \begin{align*}
0 & \geq (\pi_{11}-\pi_{11}')\lambda', \\
0 & \geq (\pi_{11}-\pi_{11}') \lambda' - (\pi_{11}-\pi_{10}')(1-\lambda'), \\ 
(\pi_{10}-\pi_{11})\lambda & \geq   (\pi_{10}-\pi_{11}') \lambda' - (\pi_{10}-\pi_{10}')(1-\lambda'), \\ 
(\pi_{10}-\pi_{11})\lambda & \geq   (\pi_{10}-\pi_{11}') \lambda'.
 \end{align*}
\end{lemma}
  Since $\pi_{11} \lambda + \pi_{10} (1-\lambda) =\pi_{1}$ and $\pi_{11}' \lambda' + \pi_{10} (1-\lambda') =\pi_{1}'$, the inequalities are equivalent to 
   \begin{align*}
0 & \geq (\pi_{11}-\pi_{11}')\lambda', \\
0 & \geq \pi_{11}-a'_2, \\ 
\pi_{10}-\pi_{1} & \geq   \pi_{10}-\pi_{1}', \\ 
\pi_{10}-\pi_{1} & \geq   (\pi_{10}-\pi_{11}') \lambda'.
 \end{align*}
  (This corresponds to $H(\pi_{11}) \geq H'(\pi_{11})$ and $H(\pi_{10}) \geq H'(\pi_{10})$ since $H$ is linear between $\pi_{11}$ and $\pi_{10}$, and, $1-\pi_{10}$ and $1-\pi_{11}$ when $\delta_2=1$.) 
  
Since $\lambda' > 0$, this is equivalent to:
  \begin{align*}
  \pi_{11} \leq \pi_{11}', \\
  \pi_{1} \leq \pi_{1}', \\ 
  (\pi_{10}-\pi_{1})(1-\lambda') = \frac{\lambda}{1-\lambda}(\pi_{1}-\pi_{11})(1-\lambda') \geq (\pi_{1}-\pi_{11}')\lambda'.
  \end{align*}
  
 The last inequality is equivalent to: 
 \begin{align*}
 \frac{1-\lambda'}{1-\lambda} \geq \frac{(1-p)\lambda'-(1-p')(1-q')}{(1-p)\lambda - (1-p)(1-q)}. 
 \end{align*} 
 Now, if $\lambda \geq \lambda'$, $(1-\lambda')/(1-\lambda) \geq 1$, it suffices that 
 \begin{align*}
 (1-p)\lambda - (1-p)(1-q) \geq (1-p)\lambda'-(1-p')(1-q') \Leftrightarrow
 \lambda-\lambda' \geq (1-q)-\frac{1-p'}{1-p}(1-q'),
 \end{align*}
 which is satisfied if $p \geq p'$ and $q \geq q'$.

Conversely, if $p \geq p'$ and $q \geq q'$, then 
 \begin{align*}
 \lambda-\lambda' & = 2(pq -p'q') - (p-p') - (q-q') \\
 &\geq 2(pq -p'q') - (p-p')q - (q-q')p \\
 & \geq p'(q -q') + q'(p-p') \geq 0,
 \end{align*}
 hence 
 \[\frac{\lambda}{1-\lambda}(\pi_{1}-\pi_{11})(1-\lambda') \geq (\pi_{1}-\pi_{11}')\lambda'. \]
 We therefore recover the known result that if both Bernouilli are more informative, i.e., $p \geq p'$ and $q \geq q'$, so is their combinations.\medskip 
 
 We have the following proposition. 
 
 \begin{proposition} The experiment $f^2$ is sufficient for $g^2$ if, and only if, $\pi_{11} \leq \pi_{11}' < \pi_{1} \leq \pi_{1}'$, and 
 \[\frac{\lambda}{1-\lambda}(\pi_{1}-\pi_{11})(1-\lambda') \geq (\pi_{1}-\pi_{11}')\lambda'. \]
 \end{proposition}
 
 \subsection{A Remark}
 So far, we have assumed that $p\geq q$. If, instead, we assume that $q>p$, then $\pi_{01} \leq  1/2 \leq \pi_{10}$, $\lambda \pi_{11} + \lambda \pi_{01}=1-q$, and the function $H_{q>p}$ becomes:
 \begin{align*}
 2* H_{q>p}(t) =
 \begin{cases}
 0 & t < \pi_{11},\\
 \lambda t - \lambda \pi_{11} & \pi_{11} \leq t < \pi_{01},\\
 t -\lambda \pi_{11} - (1-\lambda) \pi_{01} & \pi_{01} \leq t < \pi_{10}, \\ 
 t(2-\lambda) + (1-\pi_{11})\lambda- 1 & \pi_{10} \leq t <\pi_{00}, \\ 
 2 t -1 & \pi_{00} \leq t.
 \end{cases} 
 \end{align*} 
 Note that we write $H_{q>p}$ to highlight the fact that $q>p$. For completeness, recall that $H_{p >q}$ is given by: 
 \begin{align*}
 2* H_{p>q}(t) =
 \begin{cases}
 0 & t < \pi_{11},\\
 \lambda t - \lambda \pi_{11} & \pi_{11} \leq t < \pi_{10},\\
 t -\lambda \pi_{11} - (1-\lambda) \pi_{10} & \pi_{10} \leq t < \pi_{01}, \\ 
 t(2-\lambda) + (1-\pi_{11})\lambda- 1 & \pi_{01} \leq t <\pi_{00}, \\ 
 2 t -1 & \pi_{00} \leq t.
 \end{cases} 
 \end{align*} 
 In particular, for the inequality $H_{q>p}(\pi_{01}) \geq H'_{p'>q'}(\pi_{01})$ to be satisfied, we need 
 \begin{align*}
 \lambda(\pi_{01} - \pi_{11})  \geq \pi_{01} - \lambda' \pi_{11}' - (1-\lambda')\pi_{10}'.
 \end{align*}
 The inequality is equivalent to: 
 \begin{align*}
1 -p' = \lambda' \pi_{11}' - (1-\lambda')\pi_{10}' \geq \lambda \pi_{01} + (1-\lambda) \pi_{11} = 1-q,
 \end{align*}
 hence we need $q \geq p'$.

 Similarly, for the inequality $H_{p>q}(\pi_{10}) \geq H'_{q'>p'}(\pi_{10})$, we need 
 \begin{align*}
  \lambda(\pi_{10} - \pi_{11})  \geq \pi_{10} - \lambda' \pi_{11}' - (1-\lambda')\pi_{01}'. 
 \end{align*} 
 The inequality is equivalent to: 
  \begin{align*}
1 -q' = \lambda' \pi_{11}' - (1-\lambda')\pi_{01}' \geq \lambda \pi_{10} + (1-\lambda) \pi_{11} = 1-p,
 \end{align*}
 hence we need $p \geq q'$. 
 
 From our comparison of the cases $p \geq q$ and $p' \geq q'$ (and the symmetric case $q \geq p$ and $q' \geq p'$), the following proposition immediately follows: 
 
 \begin{proposition} Let $f^2$ (resp., $g^2$) be an experiment consisting of two independent Bernouilli experiments with parameters $(p,q)$ (resp., $(p',q')$). 
 The experiment $f^2$ is sufficient for the experiment $g^2$ only if $\max(p,q) \geq \max(p',q')$. 
 \end{proposition}
 
 This explains the necessity of $p \geq p'$ in our main characterization since we either need $f^1$ to be sufficient for $g^1$ or $f^2$ to be sufficient for $g^2$ for $f$ to be $\delta$-sufficient for $g$. The two conditions require that $p \geq p'$. 

\subsection{Proof of Proposition \ref{prop:ex-controlled}}\label{app:ex-controlled}
The proof is constructive. For all $z$, for all $k$, for all $\hat{k}$,  define 
\begin{align*}
\gamma(z|\vec{z}) &= \delta_1\frac{\beta_1}{\alpha_1+ \delta_2\sum_{k,\hat{z}}\alpha^k_2\xi(k|\hat{z})\gamma(\hat{z}|\emptyset)},\\
\gamma(k, (z,\emptyset)|\vec{z}) & = \delta_2\frac{\beta_1\xi(k|z)}{\alpha_1+ \delta_2\sum_{k,\hat{z}}\alpha^k_2\xi(k|\hat{z})\gamma(\hat{z}|\emptyset)},\\
\gamma(k,(\emptyset, z)|\vec{z}) & = \delta_2 \frac{\beta^k_2\xi(k|\emptyset)}{\alpha_1+ \delta_2\sum_{k,\hat{z}}\alpha^k_2\xi(k|\hat{z})\gamma(\hat{z}|\emptyset)}, \\
\gamma(\emptyset|\vec{z}) & = \frac{\delta_1(1-\beta_1)}{\delta_1(1-\beta_1) + \delta_2 \sum_{k}\beta_{\emptyset}^k\xi(k|\emptyset)}\left(1-\frac{\beta_1+ \delta_2\sum_{k}\beta^k_2\xi(k|\emptyset)}{\alpha_1+ \delta_2\sum_{k,\hat{z}}\alpha^k_2\xi(k|\hat{z})\gamma(\hat{z}|\emptyset)}\right),\\
\gamma(k, (\emptyset,\emptyset)|\vec{z}) & = \frac{\delta_2 \sum_{k}\beta_{\emptyset}^k\xi(k|\emptyset)}{\delta_1(1-\beta_1) + \delta_2 \sum_{k}\beta_{\emptyset}^k\xi(k|\emptyset)} \left(1-\frac{\beta_1+ \delta_2\sum_{k}\beta^k_2\xi(k|\emptyset)}{\alpha_1+ \delta_2\sum_{k,\hat{z}}\alpha^k_2\xi(k|\hat{z})\gamma(\hat{z}|\emptyset)}\right), \\
\end{align*}
where $ \vec{z}$ is either $(z), (\hat{k}, (z,\emptyset))$, or $(\hat{k},(\emptyset, z))$. Similarly, let
\begin{align*}
\gamma(\emptyset|\emptyset) &= \frac{\delta_1(1-\beta_1)}{\delta_1(1-\beta_1) + \delta_2 \sum_{k}\beta_{\emptyset}^k\xi(k|\emptyset)}, \\
\gamma(k, (\emptyset,\emptyset)|\emptyset) & = \frac{\delta_2 \beta_{\emptyset}^k \xi(k|\emptyset)}{\delta_1(1-\beta_1) + \delta_2 \sum_{k}\beta_{\emptyset}^k\xi(k|\emptyset)}, \\
\gamma(\emptyset|\hat{k}, (\emptyset,\emptyset)) & = \frac{\delta_1(1-\beta_1)}{\delta_1(1-\beta_1) + \delta_2 \sum_{k}\beta_{\emptyset}^k\xi(k|\emptyset)}, \\
\gamma(k, (\emptyset,\emptyset)|\hat{k}, (\emptyset,\emptyset)) & = \frac{\delta_2 \beta_{\emptyset}^k \xi(k|\emptyset)}{\delta_1(1-\beta_1) + \delta_2 \sum_{k}\beta_{\emptyset}^k\xi(k|\emptyset)},
\end{align*}
for all $k$, for all $\hat{k}$. The key observation to make is that the garbling is independent of the realized control $\hat{k}$ and, moreover, the signal $z$ is never garbled (i.e., if $z$ is observed, the support of $\gamma$ does not include  $z' \neq z$). 
\medskip

Clearly, the constructed garbling is well-defined if
\begin{align*}
\beta_1+ \delta_2\sum_{k}\beta^k_2\xi(k|\emptyset) & \leq \alpha_1+ \delta_2\sum_{k,\hat{z}}\alpha^k_2\xi(k|\hat{z})\gamma(\hat{z}|\emptyset)\\
& =  \alpha_1+ \delta_2\sum_{k}\alpha^k_2\xi(k|\emptyset) \\
& + \delta_2 \sum_{k,\hat{k}}\left[\xi(k|\hat{k},(\emptyset,\emptyset))-\xi(k|\emptyset)\right]\frac{\delta_2 \beta_0^{\hat{k}} \xi(\hat{k}|\emptyset)}{\delta_1(1-\beta_1) + \delta_2 \sum_{k}\beta_0^k\xi(k|\emptyset)} \\
&  = \alpha_1+ \delta_2\sum_{k}\alpha^k_2\xi(k|\emptyset).
\end{align*}
Since $\alpha_1 + \delta_2 \alpha_2^k \geq \beta_1 + \delta_2 \beta_2^k$ for all $k$, and $\xi(k|\emptyset)\geq 0$, we have
\begin{align*}
(\alpha_1 + \delta_2 \alpha_2^k)\xi(k|\emptyset) \geq (\beta_1 + \delta_2 \beta_2^k)\xi(k|\theta), 
\end{align*}
for all $k$. Summing over $k$ gives the desired result. 
\medskip

Lastly, because
\begin{align*}
1- \beta_1 - \delta_2 \sum_{k}\beta_2^k \xi(k|\theta) = 1- \beta_1 - \sum_{k}(1-\beta_1 -\beta_{\emptyset}^k) \xi(k|\theta) = \delta_1 (1-\delta_1) + \delta_2 \sum_{k}\beta_{\emptyset}^k\xi(k|\emptyset),
\end{align*}
it is routine to verify that the constructed garbling indeed satisfies the condition for $\delta$-sufficiency. This completes the proof.

\bibliographystyle{ecta}
\bibliography{references-renou-venel.bib}
\end{document}